\documentclass[a4paper,USenglish,cleveref, thm-restate]{lipics-v2021}
\usepackage{complexity}
\usepackage{microtype}
\usepackage[defaultlines=3,all]{nowidow}
\usepackage{tikz}
\usepackage{xspace}

\usepackage{todonotes}

\newcommand{\Isf}{\ensuremath{\mathsf{I}}\xspace}

\usepackage{pifont}
\newcommand{\xmark}{\ding{55}}
\newcommand{\cmark}{\ding{51}}%
\newcommand{\NO}{\textcolor{red}{\xmark}}
\newcommand{\YES}{\textcolor{green!60!black}{\cmark}}

\pdfoutput=1 %uncomment to ensure pdflatex processing (mandatatory e.g. to submit to arXiv)
\hideLIPIcs  %uncomment to remove references to LIPIcs series (logo, DOI, ...), e.g. when preparing a pre-final version to be uploaded to arXiv or another public repository

\usetikzlibrary{arrows, arrows.meta, automata,decorations.pathreplacing,calligraphy,fit,shapes}

\addtocounter{footnote}{-1}
\title{Deterministic Parikh Automata on Infinite Words} %TODO Please add

\author{Mario Grobler}{University of Bremen, Bremen, Germany}{grobler@uni-bremen.de}{https://orcid.org/0000-0001-8103-6440}{}%TODO mandatory, please use full name; only 1 author per \author macro; first two parameters are mandatory, other parameters can be empty. Please provide at least the name of the affiliation and the country. The full address is optional. Use additional curly braces to indicate the correct name splitting when the last name consists of multiple name parts.

\author{Sebastian Siebertz}{University of Bremen, Bremen, Germany}{siebertz@uni-bremen.de}{https://orcid.org/0000-0002-6347-1198}{}

\authorrunning{M. Grobler and S. Siebertz} %TODO mandatory. First: Use abbreviated first/middle names. Second (only in severe cases): Use first author plus 'et al.'

\Copyright{Mario Grobler and Sebastian Siebertz} %TODO mandatory, please use full first names. LIPIcs license is "CC-BY";  http://creativecommons.org/licenses/by/3.0/


 %TODO mandatory; please add comma-separated list of keywords

\category{} %optional, e.g. invited paper

\relatedversion{} %optional, e.g. full version hosted on arXiv, HAL, or other respository/website
\acknowledgements{}%optional

\nolinenumbers %uncomment to disable line numbering

\EventEditors{Aniello Murano and Alexandra Silva}
\EventNoEds{2}
\EventLongTitle{32nd EACSL Annual Conference on Computer Science Logic (CSL 2024)}
\EventShortTitle{CSL 2024}
\EventAcronym{CSL}
\EventYear{2024}
\EventDate{February 19--23, 2024}
\EventLocation{Naples, Italy}
\EventLogo{}
\SeriesVolume{288}
\ArticleNo{4}
\renewcommand{\epsilon}{\varepsilon}
\renewcommand{\phi}{\varphi}

\newcommand{\ie}{i.\,e.\xspace}
\newcommand{\eg}{e.\,g.\xspace}

\newcommand{\Amc}{\ensuremath{\mathcal{A}}\xspace}
\newcommand{\Bmc}{\ensuremath{\mathcal{B}}\xspace}

\newcommand{\Fmc}{\ensuremath{\mathcal{F}}\xspace}

\newcommand{\Kmc}{\ensuremath{\mathcal{K}}\xspace}

\newcommand{\Mmc}{\ensuremath{\mathcal{M}}\xspace}

\newcommand{\Omc}{\ensuremath{\mathcal{O}}\xspace}
\newcommand{\Pmc}{\ensuremath{\mathcal{P}}\xspace}

\newcommand{\Nbb}{\ensuremath{\mathbb{N}}\xspace}
\newcommand{\Nbbinfty}{\ensuremath{\Nbb_\infty}\xspace}
\newcommand{\Zbb}{\ensuremath{\mathbb{Z}}\xspace}

\newcommand{\bbf}{\ensuremath{\mathbf{b}}\xspace}

\newcommand{\ebf}{\ensuremath{\mathbf{e}}\xspace}
\newcommand{\ibf}{\ensuremath{\mathbf{i}}\xspace}
\newcommand{\pbf}{\ensuremath{\mathbf{p}}\xspace}

\newcommand{\ubf}{\ensuremath{\mathbf{u}}\xspace}
\newcommand{\vbf}{\ensuremath{\mathbf{v}}\xspace}

\newcommand{\0}{\ensuremath{\mathbf{0}}\xspace}
\newcommand{\1}{\ensuremath{\mathbf{1}}\xspace}

\begin{document}

\maketitle

\begin{abstract}
    Various variants of Parikh automata on infinite words have recently been introduced in the literature. 
    However, with some exceptions only their non-deterministic versions have been considered. 
    In this paper we study the deterministic versions of all variants of Parikh automata on infinite words that have not yet been studied.
    We compare the expressiveness of the deterministic models and investigate their closure properties and decision problems with applications to model checking. 
    The model of deterministic limit Parikh automata turns out to be most interesting, as it is the only deterministic Parikh model generalizing the $\omega$-regular languages, the only deterministic Parikh model closed under the Boolean operations and the only deterministic Parikh model for which all common decision problems are decidable. 
\end{abstract}
\section{Introduction}
Finite automata operating on infinite words find many applications in the world of formal languages, logic, formal verification, games, and many more.
In many scenarios, non-determinism adds expressiveness or succinctness to their deterministic counterparts. However, this often comes at the price of important decision problems becoming hard to solve, or even undecidable~\cite{ClarkeHandbook,thomas2002automata}.
This is not only true for the most common model of finite automata operating on infinite words, namely Büchi automata (which recognize the $\omega$-regular languages), but also for various extensions leaving the $\omega$-regular realm.
These include Parikh automata~(PA) on infinite words, which gained a lot of attention just recently~\cite{grobler2023remarks, infiniteZimmermann, msobapa}. Such automata are equipped with a constant number of counters (positive integers), and the counter values can be checked to satisfy some linear constraints.
Several acceptance conditions for Parikh \mbox{automata} on infinite words have been considered in the literature, including deterministic and non-deterministic reachability~PA, safety~PA, Büchi~PA and co-Büchi~PA \cite{infiniteZimmermann}, as well as non-deterministic limit~PA, reachability-regular~PA, and strong and weak reset~PA \cite{grobler2023remarks}.
For all of these models, except safety~PA and co-Büchi~PA, testing emptiness is \coNP-complete and universality is undecidable \cite{grobler2023remarks, infiniteZimmermann}. In contrast, for safety~PA and co-Büchi~PA, both problems are undecidable; however, requiring determinism yields a \coNP-complete universality problem~\cite{infiniteZimmermann}.

This motivates the study of the deterministic variants of the remaining models, namely deterministic limit~PA, reachability-regular~PA, strong reset~PA, and weak reset~PA.
In this paper we investigate their expressiveness, closure properties, and common decision problems.
Grobler et al.~\cite{grobler2023remarks} have shown that almost all non-deterministic variants of the aforementioned models (the exception being safety~PA and co-Büchi~PA again) form a hierarchy: \begin{align*}
    \text{reachability PA } & \subseteq \text{ reachability-regular~PA } = \text{ limit~PA} 
    \\ 
    & \subseteq \text{ Büchi~PA } \subseteq  
    \text{ weak reset~PA } = \text{ strong reset~PA.}
\end{align*}

\pagebreak
\noindent First, we show that this is not the case for their deterministic variants. While 

\vspace{-3mm}
\[\text{deterministic strong reset~PA $\subseteq$ deterministic weak reset~PA}\] 
and 
\begin{align*}
    \text{
deterministic} & \text{ reachability~PA} \\ &\subseteq \text{ deterministic reachability-regular~PA} \\ & \subseteq \text{deterministic weak reset~PA} 
\end{align*}
still holds, all other models become pairwise incomparable.
Furthermore, we show that among all studied deterministic models only deter\-ministic limit~PA generalize Büchi automata in the sense that they recognize all $\omega$-regular languages.

Deterministic limit~PA also shine in the light of closure properties: as we show, among all studied~PA operating on infinite words (the deterministic variants as well as the non-deterministic ones) they are the only ones closed under union, intersection and complement.

\begin{table}[h]
\begin{center}
\begin{tabular}{l|ccc}
     & $\cup$ & $\cap$ & $\bar{\phantom{a}}$ \\[0.5mm]
     \hline\\[-2.5mm]
    deterministic limit~PA& \YES & \YES & \YES \\[1mm]
    deterministic reachability-regular~PA & \NO & \NO & \NO \\[1mm]
    deterministic weak reset~PA& \NO & \NO & \NO \\[1mm]
    deterministic strong reset~PA & \NO & \NO & \NO \\[1mm]
\end{tabular}

\smallskip
\caption{\centering Closure properties.}
\label{tab:closure}
\end{center}
\end{table}

This benefit also yields decidable decision problems. In contrast to the other models that were studied in~\cite{grobler2023remarks,infiniteZimmermann}, emptiness and universality are decidable for deterministic limit~PA. We show that also strong reset~PA benefit from determinism: although having bad closure properties their universality problems becomes decidable. However, as we show, for deterministic reachability-regular~PA and deterministic weak reset~PA the universality problem remains undecidable.

\begin{table}[h]
\begin{center}
\begin{tabular}{l|ccc}
      & $L = \varnothing?$ & $uv^\omega \in L?$ & $L = \Sigma^\omega?$ \\[0.5mm] \hline\\[-2.5mm]
    deterministic limit~PA & \coNP-complete & \NP-complete & decidable, $\Pi_2^\P$-hard  \\[1mm]
    deterministic reachability-regular PA & \coNP-complete & \NP-complete & undecidable  \\[1mm]
    deterministic weak reset~PA & \coNP-complete & \NP-complete & undecidable \\[1mm]
    deterministic strong reset~PA & \coNP-complete & \NP-complete & $\Pi_2^\P$-complete \\[1mm]\end{tabular}

    \smallskip
\caption{\centering Decision problems.}
\label{tab:decision}
\end{center}
\end{table}

%\vspace{-6mm}
Finally, we study their intersection emptiness and inclusion problems, as these are important problems in order to study model checking problems, the core problems in the field of formal verification. 
In the model checking problems we are given a system~$\Kmc$ as a Kripke structure (a safety automaton) or a~PA and a specification $\Amc$ as a~PA. 
The question whether at least one computation of a Kripke structure satisfies the specification (which we call existential safety model checking and boils down to solving intersection-emptiness) has been studied in~\cite{grobler2023remarks} and is motivated by the detection of faulty computations. 
The authors of~\cite{grobler2023remarks} show that this problem is $\coNP$-complete for (non-deterministic) reset~PA and for all models that are weaker.
Similarly, the question whether all computations of a~PA satisfy the specification (which we call universal~PA model checking and boils down to solving inclusion) has been studied in~\cite{infiniteZimmermann}. 
The authors of~\cite{infiniteZimmermann} show that this problem is $\coNP$-complete for deterministic safety and co-Büchi~PA, and undecidable for their non-deterministic counterparts as well as for deterministic reachability and Büchi~PA.
We study this problem as well as the universal safety model checking problem and the existential~PA model checking problem for the remaining deterministic models. 
Again, deterministic limit~PA shine as all these problems are decidable for them. For the other models, some problems are decidable and some are not.

In the following table we list our results. Some of the questions have already been studied in the literature or follow immediately from the results of~\cite{grobler2023remarks,infiniteZimmermann}. 
We fill the missing gaps and remark that these results are the technically most demanding results of this paper.

\begin{table}[h]
\begin{center}
\begin{tabular}{l|llll}
      & \multicolumn{4}{c}{Model Checking}  \\ 
      & \multicolumn{2}{c|}{Kripke} & \multicolumn{2}{c}{PA} \\ 
      & \multicolumn{1}{c|}{$\exists$} & \multicolumn{1}{c|}{$\forall$} & \multicolumn{1}{c|}{$\exists$} & \multicolumn{1}{c}{$\forall$}  \\ \hline \\[-2.5mm]
    deterministic limit~PA & $\coNP$-c.\ \cite{grobler2023remarks} &dec., $\Pi_2^\P$-hard &$\coNP$-c.&dec., $\Pi_2^\P$-hard\\[1mm]
    deterministic reachability-regular PA &$\coNP$-c.\ \cite{grobler2023remarks} &undec.&$\coNP$-c.&undec. \\[1mm]
    deterministic weak reset~PA &$\coNP$-c.\ \cite{grobler2023remarks} &undec.&undec.&undec. \\[1mm]
    deterministic strong reset~PA &$\coNP$-c.\ \cite{grobler2023remarks}& $\Pi_2^\P$-c.& undec. &$\Pi_2^\P$-c.\\[1mm]

    deterministic reachability~PA &$\coNP$-c.\ \cite{grobler2023remarks} & undec.\ \cite{infiniteZimmermann}& $\coNP$-c.\ \cite{infiniteZimmermann} &undec.\ \cite{infiniteZimmermann}\\[1mm]
    deterministic Büchi~PA &$\coNP$-c.\ \cite{grobler2023remarks}& undec.\ \cite{infiniteZimmermann}& $\coNP$-c.\ &undec.\ \cite{infiniteZimmermann} \\[1mm]
    deterministic safety~PA &undec.\ \cite{infiniteZimmermann}& $\coNP$-c.\ \cite{infiniteZimmermann}& undec.\ \cite{infiniteZimmermann} &\coNP-c.\ \cite{infiniteZimmermann}\\[1mm]
    deterministic co-Büchi~PA &undec.\ \cite{infiniteZimmermann}& $\coNP$-c.\ \cite{infiniteZimmermann}& undec.\ \cite{infiniteZimmermann} &\coNP-c.\ \cite{infiniteZimmermann} \\[1mm]
\end{tabular}

\smallskip
\caption{\centering (Un)decidability results of model checking problems.}
\label{tab:mc}
\end{center}
\end{table}

\vspace{-6mm}

We remark that the $\Pi_2^\P$-hardness results we obtain for deterministic limit PA are not tight, that is, it remains open whether these problems can be solved in $\Pi_2^\P$. However, we are able to show that if we have the guarantee that the semi-linear sets of the PA can be complemented in polynomial time, then the $\Pi_2^\P$-hard problems for det.\ limit PA and det.\ strong reset PA become $\coNP$-complete. 
We also remark that the \coNP-completeness result for deterministic Büchi PA does not follow from~\cite{infiniteZimmermann}.

\subsection{Related work}
Parikh automata on finite words were introduced by Klaedtke and Ruess in~\cite{klaedtkeruess} in order to show that an extension of weak monadic second-order logic (MSO) with cardinality constraints is decidable.~PA turn out to be equivalent to Ibarra's reversal bounded multi-counter machines~\cite{ibarra} and Greibach's blind counter machines~\cite{greibach}, which are also known as integer vector addition systems with states ($\Zbb$-$\mathsf{VASS}$)~\cite{zvass}. Furthermore, these models can be efficiently translated into each other~\cite{logspaceconversion}.
The study of Parikh automata on infinite words was initiated just recently by Guha et al.~\cite{infiniteZimmermann}, followed by~\cite{grobler2023remarks}.
Closely related to~PA on infinite words are blind counter machines on infinite words, as introduced by Fernau and Stiebe~\cite{blindcounter}. Indeed, these machines turn out to be equivalent to Büchi~PA~\cite{grobler2023remarks}. Another model of automata with counters on infinite words are Büchi~$\mathsf{VASS}$~\cite{buchiVASScarstensen, buchiVASSValk}. 
A recent result by Baumann et al. shows that their regular separability problem (that is, the question whether an $\omega$-regular language separates two $\omega$-languages recognized by Büchi~$\mathsf{VASS}$) is decidable~\cite{buechiVASSsep}.
Another recent work shows the impact of~PA on infinite words. \mbox{Herrmann} et al.~\cite{msobapa} introduce a logic subsuming counting MSO~\cite{courcelle1} and Boolean algebra with Presburger arithmetic~\cite{bapa}, and show decidability over labeled infinite binary trees by introducing Parikh-Muller tree automata (which can be seen as an extension of reachability-regular~PA to trees).
Yet another recent work by Bergsträßer et al.~\cite{ramsey} shows that Ramsey quantifiers~\cite{ramseySem} in Presburger arithmetic can be eliminated efficiently. 
Such a quantifier basically asks for the existence of an infinite clique in the graph induced by some formulas. The authors show that their results strengthen recent results on~PA on infinite words. We use these results on Ramsey quantifiers to give an explicit application, namely that intersection emptiness of two Büchi~PA is decidable, and hence their existential~PA model checking problem.

% \subsection{Organization of the paper} 
% The paper is organized as follows. We start with the preliminaries in \Cref{sec:prelims} and give formal definitions of the studied~PA on infinite words in \Cref{sec:models}. We study their expressiveness in \Cref{sec:expressiveness}, their closure properties in \Cref{sec:closure} and their common decision problems with applications to model checking in \Cref{sec:decision}. We conclude in \Cref{sec:conclusion}.
\section{Preliminaries}
\label{sec:prelims}

\subsection{Finite and infinite words}
We write $\Zbb$ for the set of all integers and $\Nbb$ for the set of non-negative integers including~$0$. Furthermore, let $\Nbbinfty = \Nbb \cup \{\infty\}$. Throughout the paper we use the variables $c,d,k,\ell,m,n,z$ to denote positive integers and we will tacitly assume this without explicit mention. Let~$\Sigma$ be an alphabet, \ie, a finite non-empty set and let~$\Sigma^*$ be the set of all finite words over $\Sigma$. 
For a word $w \in \Sigma^*$, we denote by $|w|$ the length of $w$, and by $|w|_a$ the number of occurrences of the letter $a \in \Sigma$ in~$w$. 
We write $\varepsilon$ for the empty word of length~$0$.

An \emph{infinite word} over an alphabet $\Sigma$ is a function $\alpha : \Nbb \setminus \{0\} \rightarrow \Sigma$. We often write~$\alpha_i$ instead of~$\alpha(i)$. 
Thus, we can understand an infinite word as an infinite sequence of symbols $\alpha = \alpha_1\alpha_2\alpha_3\ldots$ For $m \leq n$, we abbreviate the finite infix $\alpha_m \ldots \alpha_n$ by $\alpha[m:n]$. 
We denote by $\Sigma^\omega$ the set of all infinite words over $\Sigma$. 
We call a subset $L \subseteq \Sigma^\omega$ an \emph{$\omega$-language}. 
Moreover, for $L \subseteq \Sigma^*$, we define $L^\omega = \{w_1w_2\dots \mid w_i \in L \setminus \{\varepsilon\}\} \subseteq \Sigma^\omega$. 
We call a non-empty $\omega$-language $L \subseteq \Sigma^\omega$ \emph{ultimately periodic} if it contains an infinite word of the form $uv^\omega$ for finite words $u,v \in \Sigma^*$. We refer to such infinite words also as \emph{ultimately periodic}.

\subsection{\boldmath Regular and \texorpdfstring{$\omega$}{omega}-regular languages}

A \emph{non-deterministic finite automaton} (NFA) is a tuple $\Amc = (Q, \Sigma, q_0, \Delta, F)$, where~$Q$ is a finite set of states, $\Sigma$ is the input alphabet, $q_0 \in Q$ is the initial state, $\Delta \subseteq Q \times \Sigma \times Q$ is the set of transitions and $F \subseteq Q$ is the set of accepting states. 
We call~$\Amc$ \emph{deterministic} if for every pair $(p, a) \in Q \times \Sigma$ there is exactly one transition of the form $(p, a, q) \in \Delta$ for some $q \in Q$.
A \emph{run} of $\Amc$ on a word $w = w_1 \ldots w_n\in \Sigma^*$ is a (possibly empty) sequence of transitions $r = r_1 \ldots r_n$ with $r_i = (p_{i-1}, w_i, p_i)\in \Delta$ such that~$p_0=q_0$. 
We say $r$ is \emph{accepting} if $p_n \in F$. The empty run on~$\epsilon$ is accepting if $q_0 \in F$. We define the \emph{language recognized by~\Amc} as $L(\Amc) = \{w \in \Sigma^* \mid \text{there is an accepting run of $\Amc$ on $w$}\}$. If a language $L$ is recognized by some NFA $\Amc$, we call~$L$ \emph{regular}.

A \emph{Büchi automaton} is an NFA $\Amc = (Q, \Sigma, q_0, \Delta, F)$ that takes \mbox{infinite} words as input. 
A \emph{run} of $\Amc$ on an infinite word $\alpha_1\alpha_2\alpha_3\dots$ is an infinite sequence of transitions $r = r_1 r_2 r_3 \dots$ with $r_i = (p_{i-1}, \alpha_i, p_i) \in \Delta$ such that $p_0=q_0$. 
We say $r$ is \emph{accepting} if there are infinitely many~$i$ with $p_i \in F$. 
We define the \emph{$\omega$-language recognized by~$\Amc$} as $L_\omega(\Amc) = \{\alpha \in \Sigma^\omega \mid \text{there is an accepting run of $\Amc$ on $\alpha$}\}$. 
If an $\omega$-language $L$ is recognized by some Büchi automaton $\Amc$, we call $L$ \emph{$\omega$-regular}.
If every state of a Büchi automaton~$\Amc$ is accepting, we call $\Amc$ a \emph{safety automaton}.
Similarly, a \emph{Muller automaton} is a tuple $\Amc = (Q, \Sigma, q_0, \Delta, \Fmc)$, where $Q, \Sigma, q_0$, and $\Delta$ are defined as for Büchi automata, and $\Fmc \subseteq 2^Q$ is a collection of sets of accepting states. Runs are defined as for Büchi automata, and a run~$r$ is accepting if the sets of states that appear infinitely often in~$r$ is contained in~$\Fmc$. 
Deterministic Muller automata have the same expressiveness as non-deterministic Büchi automata~\cite{mcnaughton1966testing}. However, deterministic Büchi automata are less expressive than their non-deterministic counterpart~\cite{landweber}. 
If an $\omega$-language $L$ is recognized by some deterministic Büchi automaton, we call $L$ \emph{deterministic $\omega$-regular}.
For a (finite word) language $L \subseteq \Sigma^*$, we define $\vec{L} = \{\alpha \in \Sigma^\omega \mid \alpha[1:i] \in L \text{ for infinitely many } i\}$. An $\omega$-language~$L$ is deterministic $\omega$-regular if and only if $L = \vec{W}$ for a regular language $W$~\cite{landweber}.

\subsection{Semi-linear sets}
%For some $d \geq 1$, let $\bbf \in \Nbb^d$ and $P = \{\pbf_1, \dots, \pbf_\ell\} \subseteq \Nbb^d$ be a finite set for some $\ell \geq 0$. A \emph{linear set} of dimension $d$ is a set of the form $C(\bbf, P) = \{\bbf + \pbf_1z_1 + \dots + \pbf_\ell z_\ell \mid z_1, \dots, z_\ell \in \Nbb\} \subseteq \Nbb^d$; we call $\bbf$ the \emph{base vector} and $P$ the set of \emph{period vectors}.
A \emph{linear set} (of dimension $d\geq 1$) is a set of the form $C(\bbf, P) = \{\bbf + \pbf_1z_1 + \dots + \pbf_\ell z_\ell \mid z_1, \dots, z_\ell \in \Nbb\} \subseteq \Nbb^d$, where $\bbf \in \Nbb^d$ and $P = \{\pbf_1, \dots, \pbf_\ell\} \subseteq \Nbb^d$ is a finite set of vectors for some $\ell \geq 0$. We call $\bbf$ the \emph{base vector} and $P$ the set of \emph{period vectors}.
A \emph{semi-linear set} is a finite union of linear sets.
We also consider semi-linear sets over~$\Nbbinfty^d$, that is, semi-linear sets with an additional symbol $\infty$ for infinity. As usual, addition of vectors and multiplication of a vector with a number is defined component-wise, where \mbox{$z + \infty = \infty + z = \infty + \infty =\infty$} for $z \in \Nbb$, $z \cdot \infty = \infty \cdot z = \infty$ for $z \geq 1$, and \mbox{$0 \cdot \infty = \infty \cdot 0 = 0$}.
For vectors $\ubf = (u_1, \dots, u_c)\in
\Nbbinfty^c$ and $\vbf = (v_1, \dots, v_d) \in \Nbbinfty^d$, we denote by $\ubf \cdot \vbf = (u_1, \dots, u_c, v_1, \dots, v_d) \in \Nbbinfty^{c+d}$ the \emph{concatenation of~$\ubf$ and $\vbf$}. 
We extend this definition to sets of vectors. Let $C \subseteq \Nbbinfty^c$ and $D \subseteq \Nbbinfty^d$. Then $C \cdot D = \{\ubf \cdot \vbf \mid \ubf \in C, \vbf \in D\} \subseteq \Nbbinfty^{c+d}$.
We denote by~$\0^d$ the $d$-dimensional all-zero vector, by $\1^d$ the $d$-dimensional all-one-vector, by~$\ebf^d_i$ the $d$-dimensional vector where the $i$th entry is~$1$ and all other entries are~0, and by~$\ibf^d_i$ the $d$-dimensional vector where the $i$th entry is~$\infty$ and all other entries are~0. We often drop the superscript $d$ if the dimension is clear. For all complexity results, we assume an explicit encoding of semi-linear sets in binary.

\subsection{Parikh-recognizable languages}
A \emph{Parikh automaton} (PA) is a tuple $\Amc = (Q, \Sigma, q_0, \Delta, F, C)$ where $Q$, $\Sigma$, $q_0$, and~$F$ are defined as for NFA, $\Delta \subseteq Q \times \Sigma \times \Nbb^d \times Q$, for some \mbox{$d\geq 1$}, is a finite set of \emph{labeled transitions}, and $C \subseteq \Nbb^d$ is a semi-linear set. 
We call $d$ the \emph{dimension} of $\Amc$ and interpret $d$ as a number of \emph{counters}.
%refer to the entries of a vector $\vbf$ in a transition $(p, a, \vbf, q) \in \Delta$ as \emph{counters}.
Analogously to NFA, we call $\Amc$ \emph{deterministic} if for every pair $(p,a) \in Q \times \Sigma$ there is exactly one labeled transition of the form \mbox{$(p,a,\vbf,q) \in \Delta$} for some $\vbf \in \Nbb^d$ and $q \in Q$.
A \emph{run} of $\Amc$ on a word $w = w_1 \dots w_n$ is a (possibly empty) sequence of labeled transitions $r = r_1 \dots r_n$ with $r_i = (p_{i-1}, w_i, \vbf_i, p_i) \in \Delta$ such that $p_0 = q_0$. We define the \emph{extended Parikh image} of a run $r$ as $\rho(r) = \sum_{i \leq n} \vbf_i$ (where the empty sum equals~$\0$).
We call $r$ accepting if~$p_n \in F$ and \mbox{$\rho(r) \in C$},
referring to the latter condition as the \emph{Parikh condition}. 
The \emph{language recognized by~\Amc} is \mbox{$L(\Amc) = \{w \in \Sigma^*\mid$} there is an accepting run of $\Amc$ on $w\}$. 
%If a language $L\subseteq \Sigma^*$ is recognized by some PA, then we call $L$ \emph{Parikh-recognizable}.
\section{Parikh automata on infinite words}
\label{sec:models}
In this section, we recall the acceptance conditions of Parikh automata operating on infinite words that were considered before in the literature. 
%We study their deterministic variants. 

\smallskip
Let $\Amc = (Q, \Sigma, q_0, \Delta, F, C)$ be a~PA. A run of $\Amc$ on an infinite word $\alpha = \alpha_1 \alpha_2 \alpha_3 \dots$ is an infinite sequence of labeled transitions $r = r_1 r_2 r_3 \dots$ with $r_i = (p_{i-1}, \alpha_i, \vbf_i, p_i)\in \Delta$ such that $p_0 = q_0$. 
The automata defined below differ only in their acceptance conditions, hence they are syntactically equivalent but semantically different (the notion of determinism translates directly). In the following, whenever we say that an automaton $\Amc$ accepts an infinite word~$\alpha$, we mean that there is an accepting run of $\Amc$ on~$\alpha$.
We begin with the models studied by Guha et al.~\cite{infiniteZimmermann}, who also studied the deterministic variants of the these models.

\begin{enumerate}
    \item The run $r$ satisfies the \emph{safety condition} if for every $i \geq 0$ we have $p_i \in F$ and $\rho(r_1 \dots r_i) \in C$. We call a~PA accepting with the safety condition a \emph{safety~PA}~\cite{infiniteZimmermann}. 
    We define the $\omega$-language recognized by a safety~PA $\Amc$ as 
    \[S_\omega(\Amc) = \{\alpha\in \Sigma^\omega\mid \Amc \text{ accepts }\alpha\}.\]

    \item The run $r$ satisfies the \emph{reachability condition} if there is an $i \geq 1$ such that $p_i \in F$ and $\rho(r_1 \dots r_i) \in C$. We say there is an \emph{accepting hit} in $r_i$. We call a~PA accepting with the reachability condition a \emph{reachability~PA}~\cite{infiniteZimmermann}. We define the $\omega$-language recognized by a reachability~PA $\Amc$ as 
    \[R_\omega(\Amc) = \{\alpha\in \Sigma^\omega\mid \Amc \text{ accepts }\alpha\}.\] 

    \item The run $r$ satisfies the \emph{Büchi condition} if there are infinitely many $i \geq 1$ such that $p_i \in F$ and $\rho(r_1 \dots r_i) \in C$. We call a~PA accepting with the Büchi condition a \emph{Büchi~PA}~\cite{infiniteZimmermann}. 
    We define the $\omega$-language recognized by a Büchi~PA $\Amc$ as 
    \[B_\omega(\Amc) = \{\alpha\in \Sigma^\omega\mid \Amc \text{ accepts }\alpha\}.\] 
    
    %Hence, a Büchi~PA can be seen as a stronger variant of a reachability~PA where we require infinitely many accepting hits instead of a single one.

    \item The run $r$ satisfies the \emph{co-Büchi condition} if there is $i_0$ such that for every $i \geq i_0$ we have $p_i \in F$ and \mbox{$\rho(r_1 \dots r_i) \in C$}. 
    We call a~PA accepting with the co-Büchi condition a \emph{co-Büchi PA}~\cite{infiniteZimmermann}. 
    We define the $\omega$-language recognized by a co-Büchi~PA $\Amc$ as 
    \[CB_\omega(\Amc) = \{\alpha\in \Sigma^\omega\mid \Amc\text{ accepts }\alpha\}.\]  
    %Hence, a co-Büchi~PA can be seen as a more restrictive variant of safety~PA where the safety condition needs not necessarily be fulfilled from the beginning, but from some point onwards.
\end{enumerate}

\medskip
\noindent Now we recall the models introduced in~\cite{grobler2023remarks}.

\begin{enumerate}
\setcounter{enumi}{4}
    \item The run satisfies the \emph{reachability and regularity condition} if there is an $i \geq 1$ such that $p_i \in F$ and $\rho(r_1 \dots r_i) \in C$, and there are infinitely many $j \geq 1$ such that $p_j \in F$. We call a~PA accepting with the reachability and regularity condition a \emph{reachability-regular}~PA~\cite{grobler2023remarks}. We define the $\omega$-language recognized by a reachability-regular~PA $\Amc$ as \[RR_\omega(\Amc) = \{\alpha\in \Sigma^\omega\mid \Amc\text{ accepts }\alpha\}.\] 
    %and call it reachability-regular.

    \item The run satisfies the \emph{limit condition} if there are infinitely many~$i \geq 1$ such that $p_i \in F$, and if additionally $\rho(r) \in C$, where the $j$th component of $\rho(r)$ is computed as follows. If there are infinitely many~$i \geq 1$ such that the $j$th component of $\vbf_i$ has a non-zero value, then the $j$th component of~$\rho(C)$ is~$\infty$. In other words, if the sum of values in a component diverges, then its value is set to $\infty$. Otherwise, the infinite sum yields a positive integer. We call a~PA accepting with the limit condition a \emph{limit~PA}~\cite{grobler2023remarks}.
    We define the $\omega$-language recognized by a limit~PA $\Amc$ as \[L_\omega(\Amc) = \{\alpha\in \Sigma^\omega\mid  \Amc\text{ accepts }\alpha\}.\] 

    \item The run satisfies the \emph{strong reset condition} if the following holds. Let $k_0 = 0$ and denote by $k_1 < k_2 < \dots$ the positions of all accepting states in~$r$. Then $r$ is accepting if the sequence $k_1, k_2, \dots$ is infinite and $\rho(r_{k_{i-1}+1} \dots r_{k_i}) \in C$ for all $i \geq 1$.
    We call a~PA accepting with the strong reset condition a \emph{strong reset~PA}~\cite{grobler2023remarks}. We define the $\omega$-language recognized by a strong reset~PA $\Amc$ as \[SR_\omega(\Amc) = \{\alpha\in \Sigma^\omega\mid \Amc\text{ accepts }\alpha\}.\] 

    \item The run satisfies the \emph{weak reset condition} if there are infinitely many reset positions $0 = k_0 < k_1 < k_2 \dots$ such that $p_{k_i} \in F$ and $\rho(r_{k_{i-1}+1} \dots r_{k_i}) \in C$ for all $i \geq 1$.
    We call a~PA accepting with the weak reset condition a \emph{weak reset~PA}~\cite{grobler2023remarks}.
    We define the $\omega$-language recognized by a weak reset~PA $\Amc$ as \[W\!R_\omega(\Amc) = \{\alpha\in \Sigma^\omega\mid \Amc\text{ accepts }\alpha\}.\] 
\end{enumerate}

Intuitively worded, whenever a strong reset~PA enters an accepting state, the Parikh condition \emph{must} be satisfied. Then the counters are reset.
Similarly, a weak reset~PA may reset the counters whenever there is an accepting hit, and they must reset infinitely often, too.
The equivalence of the (non-deterministic variants of the) two models is shown in~\cite{grobler2023remarks}. 
Furthermore, it is shown in~\cite{grobler2023remarks}  that (non-deterministic) limit~PA and (non-deterministic) reachability-regular are equivalent, as they recognize exactly $\omega$-languages of the form $\bigcup_i U_i V_i^\omega$, where the $U_i$ are (finite word) languages recognized by~PA, and the $V_i$ are regular languages.

Finally, we note that Guha et al.~\cite{infiniteZimmermann} assume that reachability~PA are complete, i.e., for every $(p,a)\in Q\times \Sigma$ there are $\vbf\in \Nbb^d$ and $q\in Q$ such that $(p,a,\vbf,q)\in \Delta$, as incompleteness allows to express additional safety conditions. 
We also make this assumption in order to avoid inconsistencies.
In fact, we can assume that all models are complete, as the other models can be completed by adding a non-accepting sink. Observe that their deterministic variants are always complete by definition.

\section{Closure Properties}
\label{sec:closure}

%\textcolor{red}{We can include the closure properties of the non-deterministic models.}

We now study the closure properties of the deterministic variants of the models introduced by Grobler et al.~\cite{grobler2023remarks}, that is, deterministic limit~PA, deterministic reachability-regular~PA, deterministic strong reset~PA, and deterministic weak reset~PA.

It is well known that semi-linear sets over $\Nbb^d$ are closed under complement~\cite{semilin}; see also~\cite{haase}.
Before we study deterministic limit automata we show that this is also true for semi-linear sets enriched with $\infty$. 

\begin{lemma}\label{lem:complement-semi-linear-with-infty}
    Let $C \subseteq \Nbbinfty^d$ be a semi-linear set. Then the complement $\bar{C} = \Nbbinfty^d \setminus C$ is semi-linear.
\end{lemma}
\begin{proof}
    Let $f : \Nbbinfty \rightarrow\Nbb$ be the bijection with $f(\infty) = 0$ and $f(i) = i+1$ for $i \in \Nbb$.
    We extend $f$ to vectors \mbox{$\vbf = (v_1, \dots, v_d) \in \Nbbinfty^d$} and sets of vectors $C \subseteq \Nbbinfty^d$ component-wise: $f(v_1, \dots, v_d) = (f(v_1), \dots, f(v_d))$ and $f(C) = \{f(\vbf) \mid \vbf \in C\}$. Note that \mbox{$f(C)\subseteq \Nbb^d$}. 

    Now we observe that $f(C)$ is semi-linear if and only if $C$ is semi-linear. First assume that $C$ is semi-linear. We may assume that $C = C(\bbf, P)$ is linear, as we can carry out the following procedure for every linear set individually. Assume $\bbf = (b_1, \dots, b_d)$. We define $D_\infty(\bbf) = \{i \mid b_i = \infty\}$.
    For a set $D \subseteq \{1, \dots, d\}$ with $D_\infty(\bbf) \subseteq D$ and a vector $\vbf = (v_1, \dots, v_d) \in \Nbbinfty^d$, let $\vbf^D = (v_1^D, \dots, v_d^D)$ with $v_i^D = 0$ if $i \in D$ and $v_i^D = v_i$ if $i \notin D$. 
    Furthermore, we call a subset $P' \subseteq P$ of period vectors \emph{$D$-compatible} if for every $i \in D \setminus D_\infty(\bbf)$, the set $P'$ contains at least one vector where the $i$th component is~$\infty$, and if for every $i \notin D$, the set $P'$ contains no vector where the $i$th component is $\infty$. 
    Observe that this definition ensures that for every vector $\vbf \in P'$ we have $\vbf^D \in \Nbb^d$, that is, no component in $\vbf^D$ is $\infty$. 
    Let $\Pmc^D_\infty \subseteq 2^P$ be the collection of $D$-compatible subsets of~$P$.
    %Then we have $f(C) = \bigcup_{D_\infty(\bbf) \subseteq D \subseteq \{1, \dots, d\}} \bigcup_{P \in \Pmc^D_\infty} \{\bbf^D + \1^D + \sum_{\pbf \in P} \pbf^D + \sum_{\pbf \in P} \pbf^D z_\pbf \mid z_\pbf \in \Nbb, \pbf\in P\}$ (where $\1$ denotes the all-one vector),
    Then we have $f(C) = \bigcup_{D_\infty(\bbf) \subseteq D \subseteq \{1, \dots, d\}} \bigcup_{P' \in \Pmc^D_\infty} C(\bbf^D + \1^D + \sum_{\pbf \in P'} \pbf^D, \{\pbf^D \mid \pbf \in P'\})$,
    which is semi-linear by definition.
    
    For the other direction, we may again assume that $f(C) = C(\bbf, P)$ is linear.
    Similar to above, assume $\bbf = (b_1, \dots, b_d)$ and define $D_0(\bbf) = \{i \mid b_i = 0\}$.
    For a set $D \subseteq D_0(\bbf)$ we call a subset $P' \subseteq P$ of period vectors \emph{$D$-safe} if for every $i \in D$ the set $P'$ contains no vector where the $i$th component is not 0, and if for every $i \notin D$ the set $P'$ contains at least one vector where $i$th component is not 0.
    Let $\Pmc^D_0 \subseteq 2^P$ be the collection of $D$-safe subsets of $P$.
    %For every subset $D \subseteq D'$ we define $C_D$ to be the subset of $C$ such that all vectors in $C_D$ have $\infty$ at position $i$ for all $i \in D$.
    %Similar to above, for a subset $D \subseteq D'$ let $B^D_{>0}$ be the collection of sets of period vectors such that every set contains at least one period vector with a non-zero entry at position $i$ for every $i \in D$. Furthermore, let $B_{|D}$ be the subset of period vectors with only zero entries at position $i$ for every $i \in D$. 
    Let $\ibf_D = (v_1, \dots, v_d)$ with $v_i = \infty$ if $i \in D$ and $v_i = 0$ if $i \notin D$.
    Then we have $C = \bigcup_{D \subseteq D_0(\bbf)} \bigcup_{P' \in \mathcal{P}^D_0} C(\ibf_D + \bbf - \1 + \sum_{\pbf \in P'} \pbf, P')$, which is semi-linear by definition (observe that every component in $\ibf_D + \bbf + \sum_{\pbf \in P'} \pbf$ is strictly greater 0, hence we can subtract $\1$ without getting negative; furthermore, we assume $\infty-1=\infty$).

    As semi-linear sets over $\Nbb^d$ are closed under complement, we have 
    \begin{align*}
        \text{$C$ is semi-linear} & \Leftrightarrow \text{$f(C)$ is semi-linear} \\
        & \Leftrightarrow \text{$\Nbb^d \setminus f(C)$ is semi-linear}\\
        & \Leftrightarrow \text{$f^{-1}(\Nbb^d \setminus f(C)) = \bar{C}$ is semi-linear}.\qedhere
    \end{align*}
\end{proof}

\begin{lemma}\label{lem:closure-limit}
    The class of deterministic limit~PA recognizable languages is closed under union, intersection and complement. 
\end{lemma}
\begin{proof}
    First observe that we can always assume that every state of a (deterministic) limit~PA is accepting, as we can check the existence of an accepting state being visited infinitely often in the semi-linear set. To achieve this, we introduce one new counter and increment it at every transition that points to an accepting state. In the semi-linear set we enforce that this counter is $\infty$.

    Hence, we can show the closure under union and intersection by a standard product construction. In case of union, we test whether at least one automaton has good counter values, and we can show the closure under intersection by testing whether both automata have good counter values.
    Closure under complement follows immediately from \Cref{lem:complement-semi-linear-with-infty}.
\end{proof}

Following the standard proof showing that the language $L_{a<\infty} = \{\alpha \in \{a,b\}^\omega \mid |\alpha|_a < \infty\}$ is not $\omega$-regular, we make 
%and make constantly use of 
the following observation. 

%The proof that these models do not recognize $L_{a<\infty}$ mimics the standard proof showing that this $\omega$-language is not deterministic $\omega$-regular, see \mbox{\eg~\cite{thomasinfinite}}.

\begin{observation}
\label{lem:not_all_reg}
    There is no det.\ reach-reg.\ PA, det.\ weak reset PA or det.\ strong reset PA recognizing the $\omega$-regular language $L_{a<\infty}$.
    %$ = \{\alpha \in \{a,b\}^\omega \mid |\alpha|_a < \infty\}$.
\end{observation}

Observe however, that the complement $L_{a=\infty} = \{\alpha \in \{a,b\}^\omega \mid |\alpha|_a = \infty\}$ of $L_{a<\infty}$ is recognized by all of these models.

\begin{lemma}\label{lem:non-closure-reach-reg}
    The class of deterministic reach-reg.\ PA recognizable languages is not closed under union, intersection or complement. 
\end{lemma}
\begin{proof}
    First we show non-closure under union.
    Let 
    $$L_1 = \{u c \alpha \mid u \in \{a,b,c\}^*, |u|_a = |u|_b, |\alpha|_c = \infty\}$$ and 
    $$L_2 = \{v a \beta \mid v \in \{a,b,c\}^*, |v|_b = |v|_c, |\beta|_a = \infty\}.$$
    
    Both languages are deterministic reach-reg.~PA recognizable, as witnessed by the automaton in \Cref{fig:limitunion} and the fact that $L_2$ can be obtained from $L_1$ by shuffling symbols. 
    We show that the language $L_1 \mathop{\cup} L_2$ is not deterministic reach-reg.~PA recognizable.
    Assume there is an $n$-state deterministic reach-reg.~PA $\Amc$ recognizing $L_1 \cup L_2$.
    Then there is a unique accepting run $r_1 r_2 \dots$ of $\Amc$ on $abc^\omega$. In particular, at some point the automaton verifies the Parikh condition, say after using the transition $r_i$. Let $m = \max\{n+1, i\}$ and consider the infinite word $abc^m b^{m-1} a^\omega$
    with the unique accepting run $r'_1 r'_2 \dots$ of $\Amc$. Due to determinism, we have $r_j = r'_j$ for all $j \leq m+2$; in particular, the automaton verifies the Parikh condition within this run prefix. As $m-1 \geq n$, the automaton visits a state, say $q$ twice while reading~$b^{m-1}$. Hence, we can pump this $q$-cycle and obtain an accepting run of~$\Amc$ on $abc^m b^{m-1+k} a^\omega$ for some $k > 0$, a contradiction. 
    
    For the non-closure under intersection, define $L_1 = \{\alpha \mid |\alpha[1:i]|_a = |\alpha[1:i]|_b \text{ for some }i \}$ and $L_2 = \{\alpha \mid |\alpha[1:i]|_a = |\alpha[1:i]|_c \text{ for some }i \}$.
    Suppose there is an $n$-state deterministic reach-reg.\ PA recognizing $L_1 \cap L_2$. Let $\alpha = a(a^n b^n)^{n+1} c^{n(n+1) + 1} a^\omega$. The unique run $r$ of $\Amc$ on $\alpha$ is not accepting, as $\alpha$ has no balanced $a$-$b$ prefix.
    Observe that~$\Amc$ visits at least one state twice while reading a $b^n$-block.
    Furthermore, there is a state, say~$q$, such that~$\Amc$ visits~$q$ twice while reading two different $b^n$-blocks, as there are $n+1$ different $b^n$-blocks.
    Hence, we can swap the latter~$q$-cycle to the front and obtain an infinite word, say $\alpha'$ in $L_1 \cap L_2$, with a unique accepting run $r'$ of $\Amc$. This run verifies the Parikh condition at some point. We distinguish two cases.
    If $r'$ verifies the Parikh condition before reading the first $c$, we can depump the $c^{n(n+1)+1}$-block and obtain an accepting run on an infinite word without a balanced $a$-$c$-prefix, a contradiction.
    Hence assume that~$r'$ verifies the Parikh condition after reading at least one $c$, say at position $k$. However, then we have $\rho(r[1:k]) = \rho(r'[1:k])$, and hence $\Amc$ also accepts $\alpha$, a contradiction.

    The non-closure under complement follows immediately from \Cref{lem:not_all_reg}.
    \begin{figure}
	\centering
	\begin{tikzpicture}[->,>=stealth',shorten >=1pt,auto,node distance=3.5cm, semithick]
	\tikzstyle{every state}=[minimum size=1cm]

 	\node[state, initial, initial text={}] (q0) {$q_0$};	
	\node[state, accepting] (q1) [right of=q0] {$q_1$};
	
	\path
	(q0) edge [loop above, align=center] node {$a, \begin{pmatrix}1\\0\end{pmatrix};$\ \ $b, \begin{pmatrix}0\\1\end{pmatrix}$} (q0)
	(q0) edge [bend left]  node {$c, \0$} (q1)
	(q1) edge [loop above] node {$c, \0$} (q1)
	(q1) edge [bend left, align=center] node {$a, \begin{pmatrix}1\\0\end{pmatrix}$;\ \ $b, \begin{pmatrix}0\\1\end{pmatrix}$} (q0)	
	;
	\end{tikzpicture}
    \caption{The deterministic reach-reg.\ ~PA with $C = C(\0, \{\1\})$ for $L_1 = \{u c \alpha \mid u \in \{a,b,c\}^*, |u|_a = |u|_b, |\alpha|_c = \infty\}$.}
    \label{fig:limitunion}
\end{figure}
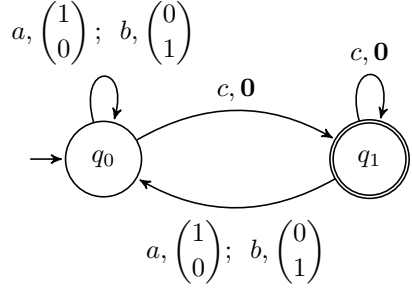
\end{proof}

\begin{lemma}\label{lem:non-closure-weak-reset}
    The class of deterministic weak reset~PA recognizable languages is not closed under union, intersection or complement. 
\end{lemma}
\begin{proof}
    We begin with the non-closure under union. The $\omega$-languages
    $L_{a=b}$ and $L_{a=c}$ (as defined in \Cref{lem:bukkiweak}) are both deterministic weak reset~PA recognizable. As shown in \Cref{lem:bukkiweak}, their union is not. 
    %Let $L_{a=b} = \{\alpha \mid |\alpha[1:i]|_a = |\alpha[1:i]|_b \text{ for $\infty$ many }i \}$ and similarly $L_{a=c} = \{\alpha \mid |\alpha[1:i]|_a = |\alpha[1:i]|_c \text{ for $\infty$ many }i \}$. \textcolor{purple}{Assume there is a deterministic weak reset~PA $\Amc$ recognizing $L_1 \cup L_2$. Consider the unique accepting run $r$ of $\Amc$ on $\alpha = (ab)^\omega$ and let $i < j$ be two positions such that $r$ resets after reading $\alpha[1:i]$ and $\alpha[1:j]$ in the same state (such a pair of positions does always exists by the infinite pigeonhole principle).
    % Now consider the infinite word $\alpha[1:i] c^{|\alpha[1:i]|_a} (ac)^\omega$ which is also accepted by $\Amc$. However, this implies that $\Amc$ also accepts $\alpha[1:j] c^{|\alpha[1:i]|_a} (ac)^\omega$, as $\Amc$ is in the same (resetting) state after reading $\alpha[1:i]$ as well as $\alpha[1:j]$, a contradiction.} MOVED UP.

    The argument for the non-closure under intersection is the same as for the non-deterministic setting~\cite{grobler2023remarks}; see also~\cite{blindcounter} and \cite{infiniteZimmermann}. Let $L_{2a = b} = \{\alpha \mid 2|\alpha[1:i]|_a = |\alpha[1:i]|_b \text{ for $\infty$ many }i \}$. Then $L_{a=b} \cap L_{2a = b}$ is not even ultimately periodic, and hence not recognized by any weak reset~PA, as they only recognize ultimately periodic $\omega$-languages~\cite{grobler2023remarks}.

    The non-closure under complement is an immediate consequence of~\Cref{lem:not_all_reg}.
    \end{proof}

\begin{lemma}\label{lem:non-closure-strong-reset}
    The class of deterministic strong reset~PA recognizable languages is not closed under union, intersection or complement. 
\end{lemma}
\begin{proof}
    We begin with the non-closure under union. Let $L = \{c^* a^n c^* b^n \mid n > 0\}^\omega$ and $L_{c=\infty} = \{\alpha \mid |\alpha|_c = \infty\}$.
    Observe that $a^n c^\omega \in L \cup L_{c=\infty}$ for every $n \geq 0$.
    Assume there is a deterministic strong reset~PA recognizing $L \cup L_{c=\infty}$. Let $n_1 \neq n_2$ be such that the unique accepting runs of $\Amc$ on $\alpha_1 = a^{n_1} c^\omega$ resp.\ $\alpha_2 = a^{n_2} c^\omega$ reset in the same state the first time they reset after reading at least one~$c$, say after reading $\alpha_1[1:i_1]$ resp.\ $\alpha_2[1:i_2]$ (with $i_1 > n_1$ and $i_2 > n_2$).
    As $a^{n_1} c^{i_1-n_1} b^{n_1} (ab)^\omega$ is also accepted by $\Amc$, the infinite word $a^{n_2} c^{i_2 - n_2} b^{n_1} (ab)^\omega$ is also accepted by $\Amc$, a contradiction.

    To show the non-closure under intersection, we use an argument similar to the non-deterministic setting. Let \mbox{$L_1 = \{a^n b^n \mid n > 0\}^\omega$} and $L_2 = \{a\} \{b^n a^{2n} \mid n > 0\}$. Then $L_1 \cap L_2$ contains only one infinite word, namely $ab a^2 b^2 a^4 b^4 \dots$. Hence $L_1 \cap L_2$ is not ultimately periodic and hence not accepted by any strong reset~PA~\cite{grobler2023remarks}.
    
    The non-closure under complement again follows from~\Cref{lem:not_all_reg}.
\end{proof}

\section{Expressiveness}
\label{sec:expressiveness}

In this section we study the expressiveness of deterministic~PA on infinite words for those models whose deterministic variants were not studied before in the literature.

First we remark that deterministic reach-reg.\ PA, deterministic limit~PA, deterministic strong reset~PA and deterministic weak reset~PA are strictly weaker than their non-deterministic counterparts. This follows immediately from their different closure properties: reach-reg.\ PA, weak reset~PA (and hence strong reset~PA) are closed under union, and limit~PA are not closed under complement~\cite{grobler2023remarks}. The authors of~\cite{grobler2023remarks} do not explicitly mention that limit~PA are not closed under complement. This however follows from their characterization of limit PA and the fact (non-deterministic finite word)~PA are not closed under complement~\cite{klaedtkeruess}.
Hence, from the results of the previous section we obtain the following corollary.

\begin{corollary}
The following strict inclusions hold.
\begin{itemize}
    \item Deterministic reach-reg.\ PA $\subsetneq$ Reach-reg.\ PA.
    \item Deterministic limit PA $\subsetneq$ Limit PA.
    \item Deterministic strong reset PA $\subsetneq$ Reset PA.
    \item Deterministic weak reset PA $\subsetneq$ Reset PA.
\end{itemize}
\end{corollary}

In the following,
when we show non-inclusions, we always give the strongest separation, \eg when showing that a deterministic strong reset PA cannot simulate a deterministic weak reset PA, we show that it can not even simulate a deterministic reachability PA, which is weaker than a deterministic weak reset PA. We refer to \Cref{fig:inclusions} for an overview of the results in this section. We simply write that a model is a strict/no subset of another model; by that we mean that the class of $\omega$-languages recognized by the first model is a strict/no subset of the class of $\omega$-languages recognized by the second model.

\begin{figure}
\centering
    \begin{tikzpicture}[%
      node distance=27mm,>=Latex,
      initial text="", initial where=below left,
      every state/.style={rectangle,rounded corners,draw=black,thin,fill=black!5,inner sep=1mm,minimum size=6mm},
      every edge/.style={draw=black,thin}
    ]
    \node[state] (reach) {rechability PA};
    \node[state, right = 2cm of reach, align=center]  (reachreg) {reachability-regular PA};
    \node[state, right= 2cm of reachreg] (wreset) {weak reset PA};
    
    \node[state,below = 1cm of reach] (reg) {det.\ $\omega$-regular};
    \node[state,right = 2cm of reg] (limit) {limit PA};
    \node[state,right = 2cm of limit] (sreset) {strong reset PA};

    \node[state,below = 1cm of reg] (safety) {safety PA};
    \node[state,right = 2cm of safety] (cobuchi) {co-Büchi PA};
    \node[state,right = 2cm of cobuchi] (buchi) {Büchi PA};

    \path[-{Latex}]
    (reach) edge (reachreg)
    (reachreg) edge (wreset)
    (sreset) edge (wreset)
    (reg) edge (limit)
    ;

    \end{tikzpicture}        
    \caption{Inclusion diagram of the studied \emph{deterministic models}. Arrows indicate strict inclusions while no connections mean incomparability.}
    \label{fig:inclusions}
\end{figure}
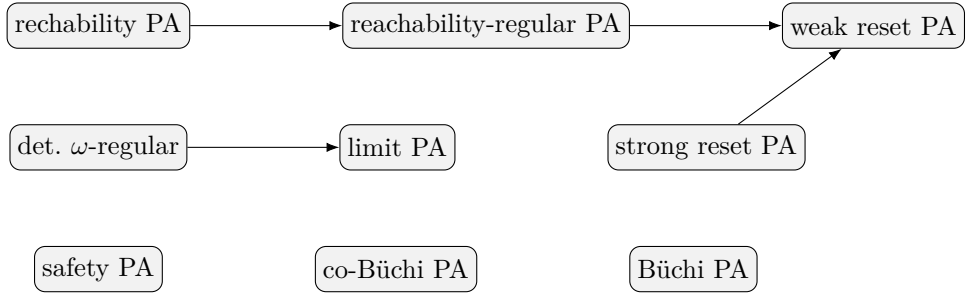

\subsection{\boldmath\texorpdfstring{$\omega$}{omega}-regular languages}

We begin by showing that every $\omega$-regular language is deterministic limit~PA recognizable. 

\begin{lemma}
%    The class of $\omega$-regular languages is a strict subclass of the class of deterministic limit~PA recognizable $\omega$-languages.
$\omega$-regular $\subsetneq$ deterministic limit PA.
\end{lemma}
\begin{proof}
    Let $L \subseteq \Sigma^\omega$ be $\omega$-regular and let $\Amc = (Q, \Sigma, q_0, \Delta, \Fmc)$ be a deterministic Muller automaton recognizing $L$.
    The idea is to construct an equivalent deterministic limit~PA $\Amc' = (Q, \Sigma, q_0, \Delta', Q, C)$ of dimension $|Q|$, where every state is accepting, while encoding the sets in $\Fmc$ into the semi-linear set $C$.
    Let $f : Q \rightarrow \{1, \dots, |Q|\}$ be a bijection associating every state with a counter. Hence, we define $\Delta' = \{(p, a, \ebf_{f(q)}^{|Q|}, q) \mid (p, a, q) \in \Delta\}$.
    %where $\ebf_{f(q)}^{|Q|}$ is the $|Q|$-dimensional vector where the $f(q)$-th component is $1$, and every other entry is 0. 
    %We define $C$ as follows.
    %For every $F \in \Fmc$, define $C_F = \{ \sum_{q \in F} \ibf_{f(q)} + \sum_{q \notin F} \ebf_{f(q)} z_q \mid z_q \in \Nbb, q \notin F\}$, where $\ibf_{f(q)}$ is the vector where the $f(q)$-th component is $\infty$, and every other entry is 0. 
    For every $F \in \Fmc$, we define $C_F = C(\sum_{q \in F} \ibf_{f(q)}, \{\ebf_{f(q)} \mid q \notin F\})$.
    That is, for every state in $F$ we expect its counter value to be $\infty$, while we expect every other counter value to be a finite number. We choose $C = \bigcup_{F \in \Fmc}C_F$ and hence obtain an equivalent deterministic limit~PA.
    
    The strictness is witnessed by the $\omega$-language $\{a^n b^n c^\omega \mid n > 0\}$, which is obviously deterministic limit~PA recognizable, but not $\omega$-regular.
\end{proof}

\Cref{lem:not_all_reg} immediately yields the following result.

\begin{corollary}
%    The class of $\omega$-regular languages is not a subclass of any of the the classes of $\omega$-languages recognized by deterministic reach-reg.~PA, deterministic Büchi~PA, deterministic strong reset~PA, or deterministic weak reset~PA.
%$\omega$-regular $\not\subseteq$ deterministic reach-reg.~PA, deterministic Büchi~PA, deterministic strong reset~PA, deterministic weak reset~PA
\mbox{}\\[1mm]
    \begin{tikzpicture}
\node[anchor=west] at (-0.3,0) {$\omega$-regular};

\node at (1.6, 0) {$\not\subseteq$};

\draw [decorate, thick,
decoration = {calligraphic brace, raise = 2pt, amplitude = 4pt,mirror}] (2.1,0.8) --  (2.1,-0.8);

\node[anchor=west] at (2.2,0.65) {Deterministic reach-reg.\ PA};
\node[anchor=west] at (2.2,0.25) {Deterministic Büchi PA};
\node[anchor=west] at (2.2,-0.25) {Deterministic strong reset PA};
\node[anchor=west] at (2.2,-0.65) {Deterministic weak reset PA};
\end{tikzpicture}
\end{corollary}
%\begin{proof}
%    None of these PA recognize $L_{a<\infty} = \{\alpha \in \{a,b\}^\omega \mid |\alpha|_a < \infty\}$, which is obviously $\omega$-regular. The proof that these models do not recognize $L_{a<\infty}$ mimics the standard proof showing that this $\omega$-language is not deterministic $\omega$-regular, see \mbox{\eg~\cite{thomasinfinite}}.
%\end{proof}
Observe however that these models generalize deterministic Büchi automata. This is not true for deterministic reachability PA, deterministic safety~PA nor deterministic \mbox{co-Büchi~PA}, as shown in the next lemma. 

\begin{lemma}\label{lem:safety-cobucki}
%The class of deterministic $\omega$-regular languages is not a subclass of any of the the classes of $\omega$-languages recognized by deterministic reach~PA, deterministic safety~PA or deterministic \mbox{co-Büchi~PA}.
%Deterministic $\omega$-regular $\not\subseteq$ deterministic reach~PA, deterministic safety~PA, deterministic \mbox{co-Büchi~PA}
\mbox{}\\[1mm]
    \begin{tikzpicture}
\node[anchor=west] at (-.7,0) {Deterministic $\omega$-regular};

\node at (3.4, 0) {$\not\subseteq$};

\draw [decorate, thick,
decoration = {calligraphic brace, raise = 2pt, amplitude = 4pt,mirror}] (3.9,0.66) --  (3.9,-0.66);

\node[anchor=west] at (4,0.47) {Deterministic reachability PA};
\node[anchor=west] at (4,0) {Deterministic safety PA};
\node[anchor=west] at (4,-0.4) {Deterministic co-Büchi PA};
\end{tikzpicture}
\end{lemma}
\begin{proof}
    As an immediate consequence of \Cref{lem:char-det-reach} (proved below) we have that no deterministic reachability~PA recognizes the deterministic \mbox{$\omega$-regular} language $a^* b^\omega$.

    The two other claims follow from~\cite[Proof of Theorem 3]{infiniteZimmermann}, where the authors have shown that (even non-deterministic) safety PA do not recognize the det.\ $\omega$-regular language \mbox{$\{a,b\}^\omega \setminus \{a\}^\omega$} and that no co-Büchi~PA recognizes the deterministic $\omega$-regular language $L_{a=\infty}=\{\alpha \in \{a,b\}^\omega \mid |\alpha|_a = \infty\}$.
\end{proof}

\subsection{Deterministic Safety PA and co-Büchi PA}

As a consequence of \Cref{lem:safety-cobucki} we obtain the following corollary. 

\begin{corollary}
\mbox{}\\[1mm]
    \begin{tikzpicture}
\node[anchor=west] at (-.7,-.2) {Deterministic reach-reg.\ PA};
\node[anchor=west] at (-.7,-0.6) {Deterministic limit PA};
%\node[anchor=west] at (0,-0.8) {Deterministic Büchi PA};
\node[anchor=west] at (-.7,-1.05) {Deterministic strong reset PA};
\node[anchor=west] at (-.7,-1.45) {Deterministic weak reset PA};

\draw [decorate,thick, 
decoration = {calligraphic brace, raise = 2pt, amplitude = 4pt}] (4,0) --  (4,-1.66);

\node at (4.5, -0.8) {$\not\subseteq$};

\draw [decorate, thick,
decoration = {calligraphic brace, raise = 2pt, amplitude = 4pt,mirror}] (5,-0.3) --  (5,-1.3);

\node[anchor=west] at (5,-0.6) {Deterministic safety PA};
\node[anchor=west] at (5,-1.1) {Deterministic co-Büchi PA};
\end{tikzpicture}
\end{corollary}

As shown in \cite{infiniteZimmermann} also deterministic reachability PA $\not\subseteq$ deterministic safety PA and deterministic reachability PA $\not\subseteq$ deterministic co-Büchi~PA as well as deterministic Büchi PA $\not\subseteq$ deterministic safety PA and deterministic Büchi PA $\not\subseteq$ deterministic co-Büchi~PA. Furthermore, the classes of deterministic safety PA and deterministic co-Büchi PA are themselves incomparable as shown in~\cite{infiniteZimmermann}. 

Vice versa, deterministic safety PA $\not\subseteq$ non-deterministic weak reset PA and deterministic co-Büchi PA $\not\subseteq$ non-deterministic weak reset PA~\cite{grobler2023remarks}. Hence, these classes are no subclasses of any of the other studied classes. 

Overall, deterministic safety PA and deterministic co-Büchi PA are incomparable with all other studied models. 

% As a consequence of this result and the results in~\cite{grobler2023remarks}, namely, that not even non-deterministic strong (equivalently non-determin\-istic weak) reset PA generalize safety PA nor co-Büchi PA, and the separations shown in~\cite{infiniteZimmermann}, we obtain that deterministic safety PA are incomparable to every other model, and determinisic co-Büchi are also incomparable to every other model.

\subsection{Deterministic Reachability PA}

We begin by characterizing the class of deterministic reachability~PA recognizable $\omega$-languages.

\begin{lemma}\label{lem:char-det-reach}
 An $\omega$-language $L$ is deterministic reachability~PA recognizable if and only if $L = U \Sigma^\omega$, where $U \subseteq \Sigma^*$ is recognized by a deterministic~PA.
\end{lemma}
\begin{proof}
Let $\Amc$ be a deterministic reachability~PA recognizing $L$. Then we have $L(\Amc) = U$. Likewise, if $\Amc$ is a~PA recognizing~$U$, then \mbox{$R_\omega(\Amc) = L$} (recall that \Amc is complete by the definition of determinism). 
\end{proof}

We have the following strict inclusion. 

\begin{lemma}
\label{lem:reachToReachReg}
%    The class of deterministic reach~PA recognizable $\omega$-lan\-guages is a strict subclass of the class of deterministic reach-reg.~PA recognizable \mbox{$\omega$-languages}.
Deterministic reachability~PA $\subsetneq$ deterministic reach-reg.\ PA.
\end{lemma}
\begin{proof}
    Let $\Amc$ be a deterministic reachability~PA. We may assume that every state of $\Amc$ is accepting, as we can project the current state into the semi-linear set. To be precise, we introduce two new counters for each state of $\Amc$, counting the number of visits and exits. Then, the current state is the (unique) state with one more visit than exit, or in case that the number of visits and exits is the same for every state, then the current state is the initial state of $\Amc$.
    As these statements can be encoded into a semi-linear set, we can assume that every state is equipped with its own semi-linear set, and can hence make every state accepting (and assign the empty set if we want to simulate a non-accepting state).
    If every state is accepting, then $\Amc$ is an equivalent deterministic reach-reg~PA.

    The strictness is witnessed e.g.\ by the $\omega$-language $\{a^nb^na^\omega\mid n>0\}$, which is deterministic reach-reg.\ PA recognizable and by \Cref{lem:char-det-reach} not deterministic reachability~PA recognizable.
\end{proof}

It remains to show the following incomparability results. 

\begin{lemma}
\label{lem:reach-no-limit}
%    The class of deterministic reach~PA recognizable $\omega$-languages is no subclass of the class of deterministic limit~PA recognizable $\omega$-languages.
Deterministic reachability PA $\not\subseteq$ deterministic limit PA.
\end{lemma}
\begin{proof}
    We show that the deterministic reachability~PA recognizable $\omega$-language $L = \{\alpha \mid |\alpha[1:i]|_a$ $= |\alpha[1:i]|_b \text{ for some }i > 0\}$ is not deterministic limit~PA recognizable. The proof is similar to the proof of Theorem~3 of (the arXiv version of)~\cite{infiniteZimmermann}.
    Assume there is an $n$-state deterministic limit~PA $\Amc$ recognizing~$L$. Consider the unique non-accepting run of $\Amc$ on $a (a^n b^n)^\omega$.
    Observe that $\Amc$ visits at least one state twice while reading a $b^n$-block, and there are at least two of the (infinitely many) \mbox{$b^n$-blocks} such that $\Amc$ visits the same state, say $q$, twice while reading them. Hence, we can shift one such \mbox{$q$-cycle} to the front and obtain the unique run on an infinite word that is in~$L$. However, this run is still non-accepting, as the extended Parikh image and number of visits of an accepting state do not change.
\end{proof}

\begin{lemma}
\label{lem:reach-no-strong-reset}
%    The class of deterministic reach~PA recognizable $\omega$-languages is no subclass of the class of deterministic strong reset~PA recognizable $\omega$-languages.
Deterministic reachability~PA $\not\subseteq$ deterministic strong reset~PA.
\end{lemma}
\begin{proof}
    We show that the deterministic reachability~PA recognizable $\omega$-language $L = \{a^n b^n \mid n \geq 1\} \cdot \{a,b\}^\omega$ is not deterministic strong reset~PA recognizable. Assume there is a deterministic strong reset~PA~$\Amc$ with $n$ states recognizing $L$. Let $\alpha = a^n b^\omega$ with unique accepting run $r = r_1 r_2 r_3 \dots$ of $\Amc$ on~$\alpha$.
    Let $f_1, f_2, \dots$ be the sequence of reset positions of~$r$ and let $i > n$ be minimal such that $i = f_{i'}$ for some $i' \geq 1$ (that is, $f_{i'}$ is the first reset position after reading a $b$).

    First observe that $i < 2n$. Assume that this is not the case. As~$\Amc$ visits at least one state twice while reading $b^n$, say state~$q$, we observe that $\Amc$ is caught in a $q \dots q$ cycle while reading $b^\omega$ due to determinism. That is, every state that is visited while reading~$b^\omega$ is already visited while reading the first $n$ many $b$s. Hence we have $i < 2n$.
    Now let $j \geq 2n$ be minimal such that $j = f_{j'}$ for some $j' > i'$ is a reset position in $r$ such that the state at position $f_{j'}$ is the same state as the one at position $f_{i'}$ (which exists by the same argument).

    Now let $\alpha' = a^n b^{j-n}a^\omega$ with unique accepting run $r' = r'_1 r'_2 r'_3 \dots$ of~$\Amc$ on $\alpha'$. 
    Observe that $\alpha[1:j] = \alpha'[1:j]$, and hence $r[1:j] = r'[1:j]$. 
    As the partial runs $r[1:i]$ and $r[1:j]$ reach the same accepting state, the run $r[1:i]r'_{j+1}r'_{j+2} \dots$ is an accepting run of~$\Amc$ on~$a^n b^{i-n} a^\omega$. 
    However, as $i-n < n$, this infinite word is not contained in $L$, a contradiction.
\end{proof}

\subsection{Deterministic Reachability-regular PA}

We begin by showing that every deterministic reach-reg.~PA (and hence every deterministic reachability~PA) can be translated into an equivalent deterministic weak reset~PA.

\begin{lemma}
\label{lem:reachRegToWeak}
%The class of deterministic reach-reg.~PA recognizable $\omega$-lan\-guages is a strict subclass of the class of deterministic weak reset~PA recognizable $\omega$-languages.
Deterministic reach-reg.\ PA $\subsetneq$ deterministic weak reset~PA.
\end{lemma}
\begin{proof}
Let $\Amc = (Q, \Sigma, q_0, \Delta, F, C)$ be a det.\ reach-reg.~PA. Let $\Amc' = (Q \cup \{q_0'\}, \Sigma, q_0', \Delta', F, C')$ be a copy of $\Amc$ with a new fresh initial state $q_0'$ inheriting all outgoing transitions of $q_0$ (observe that this modification preserves determinism). We add one new counter that is incremented at every transition leaving $q_0'$, and not modified otherwise, that is, 
\begin{align*}
\Delta' =&\ \{(p, a, \vbf \cdot 0, q) \mid (p, a, \vbf, q) \in \Delta\} 
\cup \{(q_0', a, \vbf \cdot 1, q) \mid (q_0, a, \vbf, q) \in \Delta\}.
\end{align*}

We choose $C' = C \cdot \{1\} \cup \Nbb^d \cdot \{0\}$ and obtain an equivalent weak reset~PA $\Amc'$.

The strictness is witnessed by the $\omega$-language \mbox{$\{a^n b^n \mid n > 0\}^\omega$}, which is obviously deterministic weak reset~PA-recognizable, but not even recognized by (non-deterministic) Büchi~PA, which are more expressive than reachability-regular~PA~\cite{grobler2023remarks, infiniteZimmermann}.
\end{proof}

\subsection{Deterministic Strong Reset~PA}

\begin{lemma}
%    The class of deterministic strong reset~PA recognizable $\omega$-lan\-guages is a strict subclass of the class of deterministic weak reset~PA recognizable $\omega$-lan\-guages.
Deterministic strong reset~PA $\subsetneq$ deterministic weak reset PA.
\end{lemma}
\begin{proof}
    The inclusion follows by the same argument as in the non-deterministic setting \cite{grobler2023remarks} (we use an additional counter to force the weak reset~PA to reset whenever an accepting state is visited). 

    The strictness follows from the fact that $\{a^n b^n \mid n \geq 1\} \cdot \{a,b\}^\omega$ is deterministic reachability~PA recognizable, and hence deterministic weak reset~PA recognizable (by \Cref{lem:reachToReachReg} and \Cref{lem:reachRegToWeak}), but not recognized by any deterministic strong reset~PA, as shown in \Cref{lem:reach-no-strong-reset}.
\end{proof}

\pagebreak
\begin{lemma}
%    The class of deterministic strong reset~PA recognizable $\omega$-languages is no subclass of the class of deterministic Büchi~PA recognizable $\omega$-languages.
Deterministic strong reset~PA $\not\subseteq$ deterministic Büchi~PA.
\end{lemma}
\begin{proof}
    The argument is as in \Cref{lem:reachRegToWeak}. The $\omega$-language $\{a^n b^n \mid n > 0\}^\omega$ is deterministic strong reset~PA recognizable, but there is no Büchi~PA recognizing it~\cite{infiniteZimmermann}.
    %This is witnessed by the $\omega$-language $\{a^n b^n \mid n > 0\}^\omega$, which is deterministic strong reset~PA recognizable but not even non-deterministic Büchi~PA recognizable. The latter claim follows from the fact that this language is not recognized by a blind counter automaton~\cite[Lemma 3.3]{blindcounter} and their equivalence to Büchi~PA~\cite[Lemma 21, Theorem 22]{grobler2023remarks} \textcolor{red}{we basically have this in Lemma 4.7}
\end{proof}

 \begin{lemma}
 \label{lem:strong-reset-no-limit}
 %   The class of deterministic strong reset~PA recognizable $\omega$-languages is no subclass of the class of deterministic limit~PA recognizable $\omega$-languages.
Deterministic strong reset~PA $\not\subseteq$ deterministic limit~PA.
\end{lemma}
\begin{proof}
    This follows from the previous proof as limit~PA are less expressive than Büchi~PA~\cite{grobler2023remarks}.
\end{proof}

\subsection{Deterministic Büchi~PA}
We show that $\omega$-languages recognized by deterministic Büchi~PA can be characterized in a similar way as deterministic $\omega$-regular languages.
\begin{lemma}
    An $\omega$-language $L$ is deterministic Büchi~PA recognizable if and only of $L = \vec{P}$ where $P$ is recognized by a deterministic~PA.
\end{lemma}
\begin{proof}
Let $\Amc$ be a deterministic Büchi~PA recognizing $L$ and let $\alpha \in B_\omega(\Amc)$ with accepting run $r$. As $r$ has infinitely many accepting hits by definition, we have $\alpha \in \vec{L(\Amc)}$.
Similarly, let~$\Amc$ be a deterministic~PA recognizing $P$ and let $\alpha \in \vec{P}$. As $\Amc$ is deterministic, the unique run of $\Amc$ on $\alpha$ has infinitely many accepting hits, hence we have $\alpha \in B_\omega(\Amc)$.
\end{proof}

 \begin{lemma}
 %   The class of deterministic Büchi~PA recognizable $\omega$-languages is no subclass of the class of deterministic limit~PA recognizable $\omega$-languages.
 Deterministic Büchi~PA $\not\subseteq$ deterministic limit PA.
\end{lemma}
\begin{proof}
 The proof is almost identical to the proof of \Cref{lem:reach-no-limit}, but this time we consider the $\omega$-language $L_{a=b} = \{\alpha \mid |\alpha[1:i]|_a = |\alpha[1:i]|_b \text{ for $\infty$ many }i \}$. Then we can re-use the same argument as the constructed infinite word has indeed infinitely many balanced $a$-$b$ prefixes. 
\end{proof}

 \begin{lemma}
 \label{lem:bukkiweak}
  %  The class of deterministic Büchi~PA recognizable $\omega$-languages is no subclass of the class of deterministic weak reset~PA recognizable $\omega$-languages.
  Deterministic Büchi~PA $\not\subseteq$ deterministic weak reset~PA.
\end{lemma}
\begin{proof}
Let $L_{a=b}$ be as in the last proof and similarly define $L_{a=c} = \{\alpha \mid |\alpha[1:i]|_a = |\alpha[1:i]|_c \text{ for $\infty$ many }i \}$.
We show that the deterministic Büchi~PA recognizable $\omega$-language $L_{a=b} \cup L_{a=c}$ is not deterministic weak reset~PA recognizable.
    Assume there is a deterministic weak reset~PA $\Amc$ recognizing $L_{a=b} \cup L_{a=c}$. Consider the unique accepting run $r$ of $\Amc$ on $\alpha = (ab)^\omega$ and let $i,j$ be two positions with $i+1 < j$ such that~$r$ resets after reading $\alpha[1:i]$ and $\alpha[1:j]$ in the same state (such a pair of positions does always exists by the infinite pigeonhole principle).
    Now consider the infinite word $\alpha[1:i] c^{|\alpha[1:i]|_a} (ac)^\omega$, which is also accepted by $\Amc$. However, this implies that $\Amc$ also accepts $\alpha[1:j] c^{|\alpha[1:i]|_a} (ac)^\omega$, as $\Amc$ is in the same (accepting) state after reading $\alpha[1:i]$ as well as $\alpha[1:j]$, but this infinite word is not contained in $L_{a=b} \cup L_{a=c}$, as $\alpha[1:j]$ contains at least one more $a$ than $\alpha[1:i]$, a contradiction.    
\end{proof}

We note however that the class of $\omega$-languages recognized by deterministic Büchi~PA with a \emph{linear} set form a subclass of the class of $\omega$-languages recognized by deterministic weak reset~PA with a linear set, as clarified in the following lemma.

\pagebreak
\begin{lemma}
    Let $\Amc$ be a deterministic Büchi~PA with a \emph{linear set}~$C(\bbf, P)$. Then there is an equivalent deterministic weak reset~PA.
\end{lemma}
\begin{proof}
First we observe that if $\bbf = \0$, then we have $B_\omega(\Amc) = WR_\omega(\Amc)$. To see this, let $\alpha \in B_\omega(\Amc)$ with (unique) accepting run~$r$. By Dickson's Lemma~\cite{dickson}, the run $r$ contains an infinite monotone sequence $s_1 < s_2 < \dots$ of accepting hits, that is, for all $i \geq 0$ we have $\rho(r[1:s_i]) \in C(\bbf , P)$ and for all $j > i$ we have $\rho(r[s_i+1: s_j]) \in C(\bbf, P)$. Hence, the run $r$ also satisfies the weak reset condition.
For the other direction let $\alpha \in WR_\omega(\Amc)$ with (unique) accepting run $r$ and reset positions $0 = k_0 < k_1 < k_2 \dots$. As we assume $\bbf = \0$, it is immediate that $\rho(r[1:k_i]) \in C(\bbf, P)$ for all $i \geq 1$. Hence, the run~$r$ also satisfies the Büchi condition.

Finally we argue that we can always assume that $\bbf = \0$. Indeed, we can always encode $\bbf$ into the state space of $\Amc$.
\end{proof}
\section{Decision Problems and Model Checking}
\label{sec:decision}

%\subsection{Decision Problems}
In this section, we study the following classical decision problems for~PA on infinite words. For an overview of the results in this section we refer to \Cref{tab:decision} and \Cref{tab:mc}.

\begin{itemize}
    \item Emptiness: given a PA $\Amc$, is the $\omega$-language of $\Amc$ empty?
    \item Membership: given a PA $\Amc$ and finite words $u, v$, does $\Amc$ accept $uv^\omega$?
    \item Universality: given a PA $\Amc$, does $\Amc$ accept every infinite word?
\end{itemize}

Furthermore, we study the classical model checking problem, where we are given a system~$\Kmc$ and a specification, \eg, represented as an automaton $\Amc$, and the question is whether every run of $\Kmc$ satisfies the specification, 
\ie, we ask $L(\Kmc) \subseteq L(\Amc)$, which is true if and only if $L(\Kmc) \cap \overline{L(\Amc)} = \varnothing$.
However, as complementing is often expensive or not even possible, another approach is to specify the set of all bad runs and ask whether no run of $\Kmc$ is bad, which boils down to the question is $L(\Kmc) \cap L(\Amc) = \varnothing$?
We call the first approach \emph{universal model checking} and the latter approach \emph{existential model checking}.
In our setting we assume the specification $\Amc$ to be a~PA operating on infinite words, while the system $\Kmc$ may be given as a Kripke-structure (which can be seen as a safety automaton~\cite{ClarkeHandbook}), in which case the goal is to solve \emph{safety model checking}, or also as a~PA operating on infinite words, in which case the goal is to solve \emph{PA model checking}.
Hence, we consider four problems in total, which boil down to the following decision problems.

\begin{itemize}
    \item Inclusion: given a safety automaton or a PA $\Amc_1$, and a PA $\Amc_2$, is the $\omega$-language of $\Amc_1$ a subset of the $\omega$-language of $\Amc_2$?
    \item Intersection emptiness: given a safety automaton or a PA $\Amc_1$, and a PA $\Amc_2$, is the $\omega$-language of $\Amc_1$ disjoint from the $\omega$-language of $\Amc_2$?
\end{itemize}

The techniques employed for showing \NP-completeness for non-emptiness and membership are identical and very similar to the finite word case.
However, showing undecidability and completeness results for universality and the model checking problems require more sophisticated methods. Hence, we devote most of this section to the latter problems.

\begin{lemma}
\label{lem:emptiness}
 Emptiness for deterministic limit~PA, deterministic reach-reg.~PA, deterministic weak reset~PA and deterministic strong reset~PA is $\coNP$-complete.   
\end{lemma}
\begin{proof}[Proof (sketch)]
Containment in $\coNP$ for all models is witnessed by the fact that testing non-emptiness for non-deterministic reset~PA is in \NP, and that they generalize all of these models (this follows immediately from~\cite{grobler2023remarks}).

Hardness is very similar to the finite word case~\cite{emptynp}, as we can reduce subset sum to non-emptiness for all of these models. 
\end{proof}

\begin{lemma}
 Membership for deterministic limit~PA, deterministic reach-reg.~PA, deterministic weak reset~PA and deterministic strong reset~PA is $\NP$-complete.   
\end{lemma}
\begin{proof}[Proof (sketch)]
As membership for the non-deterministic counterparts of these models is in~$\NP$~\cite{grobler2023remarks}, the upper bound follows immediately.

Hardness follows from a simple reduction from the membership problem for semi-linear sets, that is, given a semi-linear set $C \subseteq \Nbb^d$ and a vector $\vbf \in \Nbb^d$, deciding membership of~$\vbf$ in $C$. This problem is known to be $\NP$-complete, which follows immediately from the $\NP$-algorithms for integer programming~\cite{semilinUpper1, semilinUpper2} and the \NP-hardness of (a variant of) subset sum~\cite{intract, karp}; see also~\cite{haase} for a short discussion.
\end{proof}

We will now turn our attention to the universality and inclusions problems, the latter being the core of solving universal model checking.
Note that we can always reduce universality to inclusion, as an automaton $\Amc$ is universal if and only if $\Sigma^\omega$ is a subset of the $\omega$-language of $\Amc$. Hence, we show all undecidability results and lower bounds for universality and all decidability results and upper bounds for inclusion. We begin with the undecidability results.

\begin{lemma}
 Universality for det.\ reach-reg.\ PA and det.\ weak reset PA is undecidable.   
\end{lemma}
\begin{proof}
As shown in~\cite{infiniteZimmermann}, universality is already undecidable for det.\ reachability PA. As we can effectively construct equivalent det.\ reach-reg.\ PA and det.\ weak reset PA from det.\ reach.\ PA by \Cref{lem:reachToReachReg} and \Cref{lem:reachRegToWeak}, the lemma follows.
\end{proof}

Before we proceed, we show that universality for deterministic (finite word) PA is \mbox{$\Pi_2^\P$-complete}. To achieve that, we first introduce an auxiliary problem for PA and show that it is already $\Pi_2^\P$-hard even restricted to deterministic acyclic PA.
We define the \emph{irrelevance} problem for PA as follows: given a PA $\Amc = (Q, \Sigma, q_0, \Delta, F, C)$ of dimension $d$, is $\Amc$ equivalent to $\Amc^* = (Q, \Sigma, q_0, \Delta, F, \Nbb^d)$? In other words: is $C$ irrelevant for $\Amc$ in the sense that every run of $\Amc$ reaching an accepting state satisfies the Parikh condition?
In the following, we denote by $\rho(\Amc) = \{\rho(r) \mid r \text{ is an accepting run of } \Amc^*)$ and note that this set is always semi-linear as a consequence of Parikh's theorem~\cite{parikh1966context} and~\cite[Lemma 5]{klaedtkeruess}.

\begin{lemma}
\label{lem:irrelevance}
    The irrelevance problem for deterministic acyclic PA is $\Pi^\P_2$-hard.
\end{lemma}
\begin{proof}
 We reduce from the inclusion problem for integer expressions, which is known to be $\Pi^\P_2$-complete (assuming that all numbers are encoded in binary)~\cite{intexpr, polyhierarchy}. The set of integer expressions is defined as follows. Every $n \in \Nbb$ is an integer expression with $L(n) = \{n\}$. If $e_1$ and $e_2$ are integer expressions, then so are $e_1 + e_2$ and $e_1 \cup e_2$ where $L(e_1 + e_2) = \{u+v \mid u \in L(e_1), v \in L(e_2)\}$ and $L(e_1 \cup e_2) = L(e_1) \cup L(e_2)$. The inclusion problem is defined as follows: given integer expressions $e_1, e_2$, is $L(e_1) \subseteq L(e_2)$?

 We proceed as follows: first, we construct a linear set $C(\0, P_2)$ of dimension $d+1$ from~$e_2$ with the property that $n \in L(e_2)$ if and only if $n \cdot \1^d \in C(\0, P_2)$, where $d$ and $|P_2|$ depend linearly on the number of operators in $e_2$. Second, we construct a deterministic acyclic PA~$\Amc_1$ from~$e_1$ with the property $L(e_1) = \rho(\Amc_1)$. Finally, we construct a deterministic acyclic PA~$\Amc$ from~$\Amc_1$ such that $C(\0, P_2)$ is irrelevant for~$\Amc$ if and only if $L(e_1) \subseteq L(e_2)$.

 \begin{claim}
 \label{claim:intexToSet} 
  Given an integer expression $e$, one can compute in polynomial time a linear set~$C(\0, P)$ such that $n \in L(e)$ if and only if $n \cdot \1 \in C(\0, P)$.  
 \end{claim}
 \begin{claimproof}
  We construct $C(\0, P)$ inductively from $e$ and maintain the following invariant in each step: $n \in L(e)$ if and only if $n \cdot \1 \in C(\0, P)$ and for all $\pbf = (m, p_1 \dots, p_d) \in P$ we have $p_i = 1$ for at least one $1 \leq i \leq d$. We refer to \Cref{fig:intexToLinear} for an illustration.
  \begin{figure}
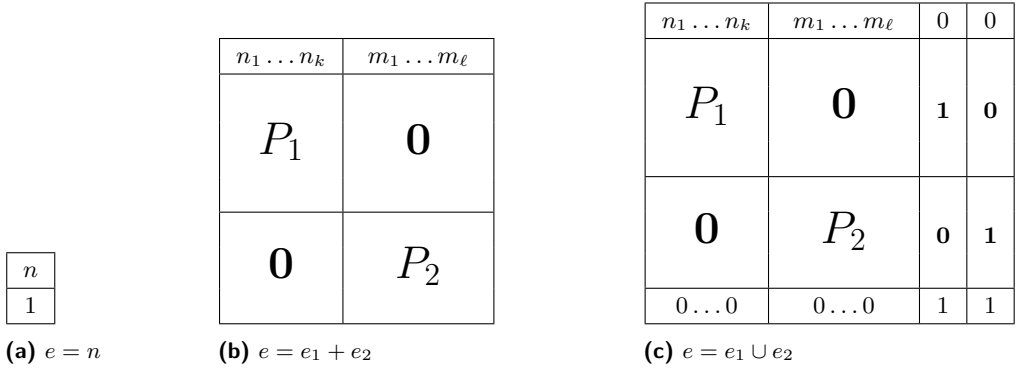

      \centering
      \begin{subfigure}[b]{0.1\textwidth}
      \begin{tabular}{|c|}
      \hline
      $n$ \\\hline $1$ \\
      \hline
      \end{tabular}
      \caption{$e = n$}
    \end{subfigure}
    \hfill
    \begin{subfigure}[b]{0.3\textwidth}
        \begin{tabular}{|c|c|}
        \hline
        $n_1 \dots n_k$ & \ $m_1 \dots m_\ell$\ {} \\\hline & \\[5pt]        
        \LARGE $P_1$ & \LARGE $\0$ \\[5pt] & \\\hline & \\[0pt]   
        \LARGE \0 & \LARGE $P_2$ \\[0pt] & \\\hline
        \end{tabular}
        \caption{$e = e_1 + e_2$}
    \end{subfigure}
    \hfill
    \begin{subfigure}[b]{0.4\textwidth}
        \begin{tabular}{|c|c|c|c|}
        \hline
        $n_1 \dots n_k$ & \ $m_1 \dots m_\ell$\ {} & 0 & 0 \\\hline & & & \\[5pt]      
        \LARGE $P_1$ & \LARGE $\0$ & \1 & \0 \\[5pt] & & &  \\\hline & & & \\[0pt]   
        \LARGE \0 & \LARGE $P_2$  & \0 & \1 \\[0pt] & & & \\\hline
        $0 \dots 0$ & $0 \dots 0$ & 1 & 1 \\\hline
        \end{tabular}
        \caption{$e = e_1 \cup e_2$}
    \end{subfigure}
      
      \caption{Illustration of the construction of a linear set from an integer expression.}
      \label{fig:intexToLinear}
  \end{figure}

 \textbf{Base Case.} If $e = n$ for $n \in \Nbb$, we choose $C(\0, P) \in \Nbb^2$ with $P = \{(n, 1)\}$. \par
 \textbf{Step.} We need to consider two cases.
 
 If $e = e_1 + e_2$, there are linear sets $C(\0, P_1) \subseteq \Nbb^{1+d_1}$ and $C(\0, P_2) + \Nbb^{1+d_2}$ for $e_1$ and~$e_2$ satisfying the invariant by assumption. We construct a linear set $C(\0, P)$ of dimension $1+d_1 + d_2$
 as follows. We pad the vectors in $P_1$ and $P_2$ with zeros to align the dimensions. Then~$P$ is the union of these vectors, \ie $P = \{n \cdot \pbf_1 \cdot \0^{d_2} \mid n \cdot \pbf_1 \in P_1\} \cup \{n \cdot \0^{d_1} \cdot \pbf_2 \mid n \cdot \pbf_2 \in P_2\}$.
 The invariant is maintained using this construction: for every integer $m + n \in L(e)$ we have $m \cdot \1^{d_1} \cdot \0^{d_2} \in C(\0, \{m \cdot \pbf_1 \cdot \0^{d_2} \mid m \cdot \pbf_1 \in P_1\})$ and $n \cdot \0^{d_1} \cdot \1^{d_2} \in C(\0, \{n \cdot \0^{d_2} \cdot \pbf_2 \mid n \cdot \pbf_2 \in P_2\})$ by assumption and construction. Hence, $(n+m) \cdot \1^{d_1 + d_2} \in C(\0, P)$. 
 Likewise, if $n \cdot \1^{d_1 + d_2} \in C(\0, P)$, we can partition the sets of used period vectors in the set of period vectors originating from $P_1$, and the set originating from $P_2$. As only the period vectors originating from $P_1$ can modify the first $d_1$ counters (ignoring the first), they yield a number $n_1 \in L(e_1)$. Similarly, the period vectors from $p_2$ yield a number $n_2 \in L(e_2)$, and hence $n = n_1 + n_2 \in L(e)$.

 If $e = e_1 \cup e_2$, there are linear sets $C(\0, P_1) \subseteq \Nbb^{1+d_1}$ and $C(\0, P_2) \subseteq \Nbb^{1+d_2}$ for $e_1$ and~$e_2$ satisfying the invariant by assumption. We construct a linear set $C(\0, P)$ of dimension $1+d_1 + d_2+1$ as
 as follows. Again, we pad the vectors in $P_1$ and $P_2$ with zeros to align the dimensions, and add an additional 0-counter. Furthermore, we consider the two vectors $\vbf_1 = 0 \cdot \1^{d_1} \cdot \0^{d_2} \cdot 1$ and $\vbf_2 = 0 \cdot \0^{d_1} \cdot \1^{d_2} \cdot 1$. 
 Then~$P$ is the union of these vectors, \ie $P = \{n \cdot \pbf_1 \cdot \0^{d_2} \cdot 0 \mid n \cdot \pbf_1 \in P_1\} \cup \{n \cdot \0^{d_1} \cdot \pbf_2 \cdot 0 \mid n \cdot \pbf_2 \in P_2\} \cup \{\vbf_1, \vbf_2\}$. The invariant is maintained using this construction: if $n \in L(e_1)$, then $n \cdot \1^{d_1 + d_2 + 1} \in C(\0, P)$ as witnessed by the following choice of period vectors. First, we can use~$\vbf_2$ to ensure that the last $d_2 + 1$ entries (including the new counter) are indeed set to one. Furthermore, the first $d_1 + 1$ entries of~$\vbf_1$ (including the first counter) are 0; hence, they do not modify the relevant counters for $n$. Hence, the containment follows from the invariant and construction. We argue analogously if $n \in L(e_2)$ using~$\vbf_1$.
 Now, if $n \cdot \1^{d_1 + d_2 + 1} \in C(\0, P)$ we observe that exactly one of the vectors~$\vbf_1$ and $\vbf_2$ must have been used, as these are the only ones that modify the last counter. If~$\vbf_1$ has been used, then no period vector originating in $P_1$ may have been used, as they all contain a one-entry by the invariant, and are hence blocked by the 1-entries in $\vbf_1$. As $\vbf_1$ does not modify the first counter, and all used period vectors originate from $P_2$, we conclude $n \in L(e_2)$. Analogously, if $\vbf_2$ has been used, we conclude $n \in L(e_1)$.

 Observe that the dimension $d$ and the size of $P$ both depend linearly on the size of $e$, hence $|C(\0, P)| \in \Omc(|e|^2)$.
 \end{claimproof}

 \begin{claim}
 \label{claim:intexToPA}
  Given an integer expression $e$, one can compute in polynomial time a deterministic acyclic PA with a single accepting state $\Amc$ such that $L(e) = \rho(\Amc)$.
 \end{claim}
 \begin{claimproof}
      \begin{figure}
         \centering
         \begin{subfigure}[b]{0.3\textwidth}
        \begin{tikzpicture}[->,>=stealth',shorten >=1pt,auto,node distance=2.5cm, semithick]
    	\tikzstyle{every state}=[minimum size=5mm]
    
     	\node[state, initial above, initial text={}] (q0) {$q_0$};	
    	\node[state, accepting] (q1) [right of=q0] {$q_1$};
    	
    	\path
    	(q0) edge  node[align=left] {$a, n$ \\ $b, n$} (q1)
    	;
    	\end{tikzpicture}
    \caption{$e = n$}
    \end{subfigure}
    \hfill
    \begin{subfigure}[b]{0.3\textwidth}
        \begin{tikzpicture}[->,>=stealth',shorten >=1pt,auto,node distance=1cm, semithick]
    	\tikzstyle{every state}=[minimum size=7mm]

     	\node[state, initial, initial text={}] (p0) {$p_0$};	
    	\node[state, accepting, dashed] (p1) [right = 6mm of p0] {$p_m$};
        \node[state, below = 18mm of p1] (q0) {$q_0$};	
    	\node[state, accepting] (q1) [left = 6mm of q0] {$q_n$};

        \node[right = -0.3mm of p0] {$\dots$};
        \node[left = -0.6mm of q0] {$\dots$};

        \node[fit = (p0) (p1), draw]{};
        \node[fit = (q0) (q1), draw]{};
    	
    	\path
        (p1) edge node[align=left] {$a, 0$ \\ $b, 0$} (q0)
    	;
    	\end{tikzpicture}
        \caption{$e = e_1 + e_2$}
    \end{subfigure}
    \hfill
    \begin{subfigure}[b]{0.3\textwidth}
        \begin{tikzpicture}[->,>=stealth',shorten >=1pt,auto,node distance=1cm, semithick]
    	\tikzstyle{every state}=[minimum size=7mm]

        \node[state, initial above, initial text={}] (q) {$q$};	
     	\node[state, above right = 8mm and 3mm of q] (p0) {$p_0$};	
    	\node[state, accepting, dashed] (p1) [right = 6mm of p0] {$p_m$};
        \node[state, below right = 8mm and 3mm of q] (q0) {$q_0$};	
    	\node[state, accepting, dashed] (q1) [right = 6mm of q0] {$q_n$};
        \node[state, accepting] (qf) [below right = 8mm and 3mm of p1] {$q_f$};

        \node[right = -0.3mm of p0] {$\dots$};
        \node[right = -0.3mm of q0] {$\dots$};

        \node[fit = (p0) (p1), draw]{};
        \node[fit = (q0) (q1), draw]{};
    	
    	\path
        (q)  edge node[align=left] {$a, 0$}                        (p0)
        (q)  edge node[align=left, below left] {$b, 0$}            (q0)
    	(p1) edge node[align=left] {$a, 0$ \\ $b, 0$}              (qf)
        (q1) edge node[align=left, below right] {$a, 0$ \\ $b, 0$} (qf)
    	;
    	\end{tikzpicture}
        \caption{$e = e_1 \cup e_2$}
    \end{subfigure}

         \caption{Illustration of the construction of a deterministic acyclic PA from an integer expression.}
         \label{fig:intexToPA}
     \end{figure}
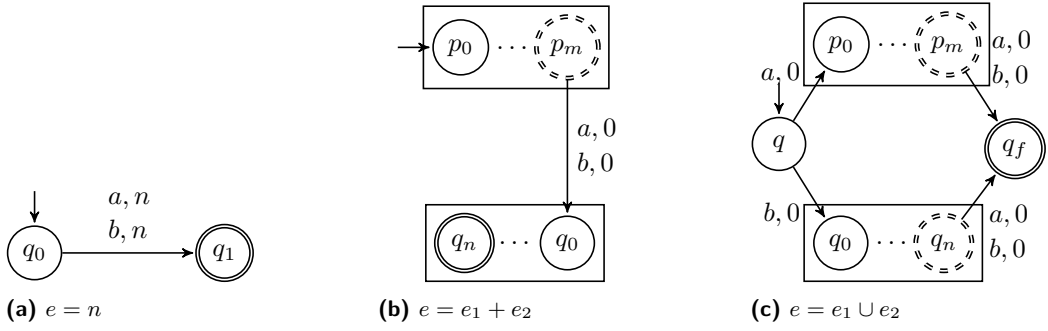
     We construct $\Amc$ over the alphabet $\{a,b\}$ of dimension $1$ with linear set $\Nbb$ inductively from $e$ and maintain the invariant in the claim in every step. We refer to \Cref{fig:intexToPA} for an illustration.

     \textbf{Base Case.} If $e = n$ for $n \in \Nbb$, the PA $\Amc$ consists of an initial state and an accepting state that are connected by an $a$-transition and a $b$-transition, both labeled with $n$.

     \textbf{Step.} We need to consider two cases.\par
     If $e = e_1 + e_2$, there are PA $\Amc_1$ and $\Amc_2$ for $e_1$ and $e_2$ satisfying the invariant by assumption. We construct a PA $\Amc$ as follows: we take the disjoint union of $\Amc_1$ and $\Amc_2$ and connect the accepting state of $\Amc_1$ with the initial state of $\Amc_2$ via an $a$-transition and a $b$-transitions, both labeled with $\0$. Finally, the accepting states of $\Amc_1$ is not accepting anymore in $\Amc$.
     
     If $e = e_1 \cup e_2$, there are PA $\Amc_1$ and $\Amc_2$ for $e_1$ and $e_2$ satisfying the invariant by assumption. We construct a PA $\Amc$ as follows: we take the disjoint union of $\Amc_1$ and $\Amc_2$ and add a fresh initial state, say $q$, and a fresh accepting state, say $q_f$. Then, we connect $q$ with the initial state of $\Amc_1$ with an $a$-transition labeled with $\0$, and we connect $q$ with the initial state of $\Amc_2$ with a $b$-transition labeled with $\0$. Similarly, we connect the accepting state of $\Amc_1$ as well as the accepting state of $\Amc_2$ with $q_f$ via $a$-transitions and $b$-transitions, all labeled with $\0$. Finally, the accepting states of $\Amc_1$ and $\Amc_2$ are not accepting anymore in $\Amc$.
 \end{claimproof}

We are now ready to prove \Cref{lem:irrelevance}. Let $\Amc_1$ be the PA for $e_1$ as constructed in the proof of \Cref{claim:intexToPA}. Similarly, let $C(\0, P_2)$ be the linear set of dimension $1+d$ for $e_2$ as constructed in the proof of \Cref{claim:intexToSet}.
Now we construct the deterministic acyclic PA $\Amc$ as follows. We start with $\Amc_1$ and pad all vectors with zeros, that is, we replace every (one-dimensional) vector $n$ in $\Amc_1$ by $n \cdot \0^d$.
Then, we add a fresh accepting state and connect it from the accepting state of $\Amc_1$ with an $a$-transition and $b$-transition, both labeled with $0 \cdot \1^d$.
Finally, the accepting state of $\Amc_1$ is not accepting anymore in $\Amc$, and we choose $C(\0, P_2)$ as the linear set of $\Amc$.
Observe that the properties of $\Amc_1$ and $C(\0, P_2)$ ensure that $C(\0, P_2)$ is irrelevant for $\Amc$ if and only if $L(e_1) \subseteq L(e_2)$. This concludes the proof. 
\end{proof}

\begin{remark*}
Recall that we assume a binary and explicit encoding of semi-linear sets. In terms of expressiveness we can equally assume that they are given as Presburger formulas. However, this drastically changes the complexity: if the semi-linear sets are given as Presburger formulas, the problem becomes $\coNEXP$-complete, as we can interreduce the problem with the $\forall^*\exists^*$-fragment of Presburger arithmetic, which is $\coNEXP$-complete~\cite{haaseNexp}.
\end{remark*}

Observe that we can easily reduce the previous problem to the universality problem for deterministic PA.
\begin{corollary}
\label{cor:uniHardnessFinite}
 Universality for deterministic PA is $\Pi_2^\P$-hard.
\end{corollary}
\begin{proof}
Let $\Amc$ be the deterministic acyclic PA with linear set $C(\0, P)$ as constructed in the previous proof. 
We construct a deterministic $\Amc'$ such that $C(\0, P)$ is irrelevant for $\Amc$ if and only if $\Amc'$ is universal.
To achieve that, we start with $\Amc$ and make every state accepting. 
Furthermore, we add an $a$-loop and a $b$-loop to the accepting state of $\Amc$, both labeled with $\0$. Finally, the semi-linear set of $\Amc'$ is $C(\0, P) \cup (\Nbb \cdot \{\0\})$.
\end{proof}

Now we show that universality for deterministic PA is in $\Pi_2^\P$, yielding completeness for universality and irrelevance by the reduction presented in the proof of the previous corollary.

\begin{lemma}
\label{lem:uniFinite}
    Universality for deterministic (finite word) PA is $\Pi_2^\P$-complete.
\end{lemma}
\begin{proof}
We show that the problem is contained in $\Pi_2^\P$ by showing that its complement is contained in $\Sigma_2^\P$. We make use of the following observation. As we can assume that every state is accepting (by encoding accepting states into the semi-linear set), the question whether a det.\ PA $\Amc$ with semi-linear set $C$ is not universal boils down to the question $\rho(\Amc) \not\subseteq C$?
If both sets are given explicitly (again assuming binary encoding), Huynh showed that this question can be decided in $\Sigma_2^\P$~\cite{huynhUpperSimplified, huynhUpper}.
However, we cannot explicitly compute the set $\rho(\Amc)$ as its size can be exponentially large in the size of~$\Amc$~\cite{parikhComplexity}. To circumvent this problem, we guess a small linear subset of $\rho(\Amc)$ that witnesses non-inclusion.

In the first step, we compute an existential Presburger formula, say $\varphi(v_1, \dots, v_d) = \exists x_1 \dots x_m \psi$ for the set $\rho(\Amc)$. This can be done in linear time using the results in~\cite{schwentikHorn}; we also refer to~\cite[Proposition III.2]{emptynp} for a short discussion on the construction.

In the second step, we use the result in~\cite[Corollary II.2]{zvassnz} essentially stating that every semi-linear set can be written as a finite union of linear sets of small bit size. To be precise, the result states that every semi-linear set $C \subseteq \Nbb^d$ can be written as $\bigcup_i C(\bbf_i, P_i)$ with $|P_i| \leq 2d \log(4d \|C\|)$, where $\|C\|$ denotes the largest absolute value that appears in a base or period vector in any linear set in the union of $C$. Hence, we guess a linear set $C(\bbf_i, P_i) \subseteq \rho(\Amc)$ of small bit size. Note however, that we do not verify at this point whether $C(\bbf_i, P_i)$ is indeed a subset of $\rho(\Amc)$.

In the third step, we use Huynh's results~\cite{huynhUpperSimplified, huynhUpper} to solve $C(\bbf_i, P_i) \not\subseteq C$ in $\Sigma_2^\P$. As one main ingredient to obtain containment $\Sigma_2^\P$, Huynh showed that every non-empty (set-theoretic) difference of two semi-linear sets contains a vector of polynomial bit size. Hence let $\vbf \in C(\bbf_i, P_i) \setminus C$ be such a vector whose bit size is bounded polynomially in $C(\bbf_i, P_i)$ and~$C$, and hence in the input size.

Finally, we need to verify that $\vbf$ is indeed a member of $\rho(\Amc)$ (recall that we did not verify that $C(\bbf_i, P_i)$ is indeed a subset of $C_\Amc$). In order to do so, we guess a valuation of the quantified variables $x_1, \dots, x_n$ (again of polynomial bit size) of $\varphi$ and verify that $\varphi(\vbf)$ is indeed satisfied under this valuation.

Overall, we conclude that non-universality for det.\ PA is in $\Sigma_2^\P$, yielding the desired result.
\end{proof}

From \Cref{lem:irrelevance} we can also conclude $\Pi_2^\P$-hardness for universality for deterministic limit and deterministic strong reset PA using a simplified variant of the reduction in \Cref{cor:uniHardnessFinite}.

\begin{corollary}
\label{cor:universalityLimit}
Universality and inclusion for deterministic limit~PA and deterministic strong reset~PA is $\Pi_2^\P$-hard.
\end{corollary}

Now we focus on the decidability results and upper bounds. Let $\Amc_1$ and $\Amc_2$ be two (deterministic limit) PA. Note that $L_\omega(\Amc_1) \subseteq L_\omega(\Amc_2)$ holds if and only if \mbox{$L_\omega(\Amc_1) \cap \overline{L_\omega(\Amc_2)} =\varnothing$}. As det.\ limit PA are effectively closed under complement and intersection, and emptiness is $\coNP$-complete (and hence decidable) for them, we obtain the following result.
\begin{corollary}
Universality and inclusion for deterministic limit PA are decidable.    
\end{corollary}
Unfortunately, we do not obtain tight bounds and conjecture that these problem are $\Pi_2^\P$-complete for them. 
\setcounter{conjecture}{0}
\begin{conjecture}
 Universality and inclusion for deterministic limit PA are $\Pi_2^\P$-complete.    
\end{conjecture}
\setcounter{conjecture}{33}
However, we make the following observation. The relatively high $\Pi_2^\P$ lower bound of universality for det.\ limit PA and det.\ strong reset PA (and also of irrelevance for finite word PA)
can be explained by the cost of (implicitly) complementing semi-linear sets. In fact, if we have the guarantee that the semi-linear set of the second PA can be complemented in polynomial time, the universality and inclusion problems become $\coNP$-complete.
We start with det.\ limit PA.

\begin{lemma}
\label{lem:limitcoNP}
Let $\Amc_1$ and $\Amc_2$ be deterministic limit PA with the guarantee that the semi-linear set of $\Amc_2$ can be complemented in polynomial time. Then the following questions are $\coNP$-complete.
\begin{enumerate}
    \item Is $L_\omega(\Amc_2) = \Sigma^\omega$?
    \item Is $L_\omega(\Amc_1) \subseteq L_\omega(\Amc_2)$?
\end{enumerate}
\end{lemma}
\begin{proof}
    Containment in $\coNP$ of both questions follows immediately from a reduction to emptiness for det.\ limit PA, as the guarantee allows us to complement det.\ limit PA in polynomial time, and we can construct the product automaton of two det.\ limit PA in polynomial time.

    The partition problem is defined as follows: given a multiset $M$ of positive integers, is there a subset $M'$ of $M$ such that $\sum_{n \in M} n = \sum_{n \in M \setminus M'} n$?. This problem is one of the classical \NP-complete problems~\cite{intract}. We reduce from its complement.
    
    Let $M = \{n_1, \dots n_k\}$ be a multiset of positive integers. We construct a det.\ limit PA~$\Amc$ over the alphabet $\{a,b\}$ of dimension 2 as follows. The state set of $\Amc$ is $\{q_0, q_1, \dots q_k\}$ where~$q_0$ is the initial state and $q_k$ is the only accepting state. For every $1 \leq i \leq k$, there is an $a$-transition from $q_{i-1}$ to ${q_i}$ labeled with $(n_i, 0)$ as well as a $b$-transition labeled with~$(0, n_i)$. Finally, we add an $a$-loop and a $b$-loop to the accepting state $q_k$, both labeled with $\0$. The semi-linear set of $\Amc$ is 
    \[C = \{(z,z') \mid z \neq z' \} = C((1,0), \{(1,0), (1,1)\}) \cup C((0,1), \{(0,1), (1,1)\}),\]
    whose size does not depend on the size of $M$. Furthermore, the complement of $C$ is 
    \[\overline{C} = \{(z,z) \mid z \in \Nbb\} = C(\0, \{(1,1\}),\]
    and can hence be computed in polynomial time.

    Now we have that $\Amc$ is universal if and only if $M$ is a negative instance of partition. To see this, observe that every prefix $w_1 \dots w_k \in \{a,b\}^k$ of every word $\alpha \in L_\omega(\Amc)$ represents a subset $M'$ of $M$ with $n_i \in M'$ if and only if $w_i = a$, and hence $n_i \in M \setminus M'$ if and only if~$w_i = b$. The semi-linear set $C$ of $\Amc$ states that $M'$ and $M \setminus M'$ are not a partition of~$M$.% Hence, the lemma follows.
\end{proof}

Now we show that this is also the case for det.\ strong reset PA, presenting in parallel the strategy that we will finally adapt to show that inclusion is $\Pi_2^\P$-complete in the general case.
We show that complements of det.\ strong reset PA recognizable $\omega$-languages are reach-reg.\ PA recognizable, and given that we can complement the semi-linear set in polynomial time, we can construct such a reach-reg.\ PA in polynomial time. Subsequently, we show how to test intersection emptiness of a det.\ strong reset PA and a reach-reg.\ PA in $\coNP$. We begin by proving the first main ingredient.

\begin{lemma}
\label{lem:complSR}
    Let $\Amc = (Q, \Sigma, q_0, \Delta, F, C)$ be a det.\ strong reset PA. Then there is a (non-deterministic) reach-reg.\ PA $\Amc'$ recognizing $\overline{SR_\omega(\Amc)}$. Furthermore, if $C$ can be complemented in polynomial time, then we can compute $\Amc'$ in polynomial time.
\end{lemma}
\begin{proof}
    We assume that every state of $\Amc$ is reachable from the initial state $q_0$ (as we can safely remove unreachable states in polynomial time).
    Observe that $\Amc$ rejects an infinite word $\alpha$ whenever one of the following two conditions is met:
    \begin{enumerate}
        \item The unique run of $\Amc$ on $\alpha$ visits every accepting state just finitely often.
        \item The unique run of $\Amc$ on $\alpha$ visits an accepting state with bad counter values.
    \end{enumerate}

    The first condition is $\omega$-regular, while the second one can be tested with a reachability~PA. 
    We begin by clarifying the first point. Testing whether there is an infinite word rejected by $\Amc$ because its unique run visits every accepting only finitely often boils down to testing whether the complement $\omega$-language of the underlying Büchi automaton of $\Amc$ is non-empty.
    Hence, let $\Bmc$ be the Büchi automaton obtained from~$\Amc$ by forgetting all vectors and let $\overline{\Bmc}$ be a Büchi automaton recognizing the complement of~$L_\omega(\Bmc)$. As $\Bmc$ is deterministic, we can construct $\overline{\Bmc}$ in polynomial time~\cite{detBuchiCompl}. Then $\overline{\Bmc}$ accepts all words rejected by $\Amc$ due to the first condition.

    Second, we show how to construct a (non-deterministic) reachability PA accepting all infinite words that are rejected by $\Amc$ due the second condition. We re-use the following idea from~\cite[Lemma 30]{grobler2023remarks}. For any two states $p, q \in Q$ let $\Amc_{p \Rightarrow q} = (Q \cup \{q_0'\}, \Sigma, q_0', \Delta_{p \Rightarrow q}, \{q\}, C)$ where $\Delta_{p \Rightarrow q} = \{(q_1, a, \vbf, q_2) \mid (q_1, a, \vbf, q_2) \in \Delta, q_1 \notin F) \cup \{(q_0', a, \vbf, q_2) \mid (p, a, \vbf, q_2) \in \Delta\}$. 
    Hence, if we interpret $\Amc_{p\Rightarrow q}$ as a finite word~PA, then it accepts all words accepted by the finite word~PA $\Amc$ when starting in $p$, ending in $q$, and not visiting an accepting state in-between. In other words, if $p,q\in F$, then $\Amc_{p \Rightarrow q}$ accepts all finite infixes that the strong reset~PA $\Amc$ may read when the last reset was in $p$, and the next reset is in $q$. %Furthermore, this construction preserves determinism (up to completeness, but we can always complete the automaton by adding a non-accepting sink).

    For every pair of accepting states $p,q \in F$, we consider the PA $\overline{\Amc}_{p \Rightarrow q}$, which is defined as~$\Amc_{p \Rightarrow q}$ but with the complement semi-linear set $\overline{C}$ of $C$. 
    Observe that we can compute~$\overline{\Amc}_{p \Rightarrow q}$ in polynomial time if $C$ can be complemented in polynomial time. Hence, $\overline{\Amc}_{p \Rightarrow q}$ accepts all finite words collecting a vector not in $C$ when starting in $p$, ending in $q$, and not visiting other accepting states in-between.
    Now let $\Amc^\circ_{p \Rightarrow q}$ the PA obtained from $\Amc$ and $\overline{\Amc}_{p \Rightarrow q}$ as follows. We start with a copy of $\Amc$ but every state is not accepting and every vector is replaced by~$\0$. 
    Then, we identify the state $p$ in $\Amc^\circ_{p \Rightarrow q}$ with the initial state of $\overline{\Amc}_{p \Rightarrow q}$, that is, we remove every outgoing transition of $p$ in $\Amc^\circ_{p \Rightarrow q}$ and replace them with the outgoing transitions of the initial state of $\overline{\Amc}_{p \Rightarrow q}$. 
    Furthermore, the accepting state $q$ of $\overline{\Amc}_{p \Rightarrow q}$ is the only accepting state of $\Amc^\circ_{p \Rightarrow q}$. Observe that $q$ has no outgoing transitions by construction. 
    Finally, we add a trivial self-loop to $q$ and choose $\overline{C}$ to be the semi-linear set of $\Amc^\circ_{p \Rightarrow q}$.
    Hence, $\Amc^\circ_{p \Rightarrow q}$ accepts all infinite words of the form $uv\beta$ with $\beta \in \Sigma^\omega$ such that $\Amc$ is in state $p$ after reading $u$, then in state $q$ after further reading $v$, and collects a vector not in $C$ while reading~$v$ and hence rejects every infinite word with prefix $uv$.
    We compute the PA $\Amc^\circ_{p \Rightarrow q}$ for every $p, q \in F$ and let $\Amc^\circ$ be the reachability PA recognizing the union of all these $\Amc^\circ_{p \Rightarrow q}$ by taking the disjoint union and connecting them with a fresh initial state (see \cite[Lemma 6]{infiniteZimmermann} for details). 
    In contrast to an iterated product construction, the presented construction allows us to compute~$\Amc^\circ$ in polynomial time, albeit not preserving determinism.

    Finally, let $\Amc'$ be a reach-reg.\ PA accepting $L_\omega(\overline{\Bmc}) \cup R_\omega(\Amc^\circ)$. Note that $\Amc'$ can be computed in polynomial time in the sizes of $\overline{\Bmc}$ and $\Amc^\circ$ by turning both automata into reach-reg.\ PA and again taking their disjoint union with a fresh initial state. Now we have $RR_\omega(\Amc') = \overline{SR_\omega(\Amc)}$.
\end{proof}

Before we proceed, we show two auxiliary lemmas.
\begin{lemma}
    Let $\Amc_1 = (Q_1, \Sigma, p_I, \Delta_1, F_1, C)$ be a deterministic strong reset~PA and $\Bmc = (Q_2, \Sigma, q_I, \Delta_2, F_2)$ be a Büchi automaton. Then $SR_\omega(\Amc_1) \cap L_\omega(\Bmc)$ is ultimately periodic.
\end{lemma}
\begin{proof}
    Assume $SR_\omega(\Amc_1) \cap L_\omega(\Bmc) \neq \varnothing$ and let $\alpha$ be an infinite word accepted by both automata, $\Amc_1$ and $\Bmc$.
    If $\alpha = uv^\omega$ for some $u, v \in \Sigma^*$, we are done. Hence assume that this is not the case.
    Let $\Amc = (Q_1 \times Q_2, \Sigma, (p_I, q_I), \Delta, F_1 \times Q_2, C)$ with $\Delta = \{((p, q), a, \vbf, (p', q')) \mid (p, a, \vbf, p') \in \Delta_1, (q, a, q') \in \Delta_2\}$ be the product automaton\footnote{We note that $\Amc$ interpreted as a strong reset~PA does not recognize $SR_\omega(\Amc_1) \cap L_\omega(\Bmc)$. Instead, it accepts all infinite words $\alpha \in SR_\omega(\Amc_1)$ such that $\Bmc$ has an infinite but not necessarily accepting run on~$\alpha$.} of $\Amc_1$ and $\Bmc$.
    As $\alpha \in SR_\omega(\Amc_1)$, the unique accepting run of $\Amc_1$ on $\alpha$, say $(p_0, \alpha_1, \vbf_1, p_1)$ $(p_1, \alpha_2, \vbf_2, p_2) \dots$ with $p_0 = p_I$, is accepting. Likewise, as $\alpha \in L_\omega(\Bmc)$, there is an accepting run $(q_0, \alpha_1, q_1) (q_1, \alpha_2, q_2) \dots$ with $q_0 = q_I$ of $\Bmc$ on $\alpha$.
    Hence, $r = ((p_0, q_0), \alpha_1, \vbf_1, (p_1, q_1)) ((p_1, q_1), \alpha_2, \vbf_2, (p_2, q_2)) \dots$ is a run of $\Amc$ on $\alpha$
    with the following properties:
    
    \begin{itemize}
        \item there is $p_f \in F_1$ such that for infinitely many $i$ we have $p_i = p_f$. Let $f_1, f_2, \dots$ denote the positions of all occurrences of a state of the form $(p_f, \cdot)$ in $r$.
        \item there is $q_f \in F_2$ such that for infinitely many $i$ we have $q_i = q_f$. 
        %Let $f'_1, f'_2, \dots$ denote the positions of all occurrences of a state of the form $(\cdot, q_f)$ in $r$.
    \end{itemize}
    
    Let $j \geq f_1$ be minimal such that $q_j = q_f$.
    Now let $k \leq j$ be maximal such that $k = f_\ell$ for some $\ell \geq 1$.
    Then $r[1:f_\ell]r[f_\ell+1 : f_{\ell+1}]^\omega$ is an accepting run of $\Amc$ on an ultimately periodic word, say $uv^\omega$, that visits an accepting state of $\Bmc$ infinitely often. Hence $uv^\omega \in SR_\omega(\Amc_1) \cap L_\omega(\Bmc)$.
\end{proof}
\begin{lemma}
\label{lem:interSRBukki}
Let $\Amc_1 = (Q_1, \Sigma, p_0, \Delta_1, F_1, C)$ be a deterministic strong reset~PA and let $\Bmc = (Q_2, \Sigma, q_0, \Delta_2, F_2)$ be a Büchi automaton. The question $SR_\omega(\Amc_1) \cap L_\omega(\Bmc) = \varnothing$ is \coNP-complete.
\end{lemma}
\begin{proof}
The lower bound follows immediately from the $\coNP$-hardness of emptiness; hence, we focus on the containment in $\coNP$ by showing that testing intersection non-emptiness is in \NP.
By the previous lemma it is sufficient to check the existence of an ultimately periodic word.
%Let $\Amc_1 = (Q_1, \Sigma, p_0, \Delta_1, F_1, C)$ and $\Bmc = (Q_2, \Sigma, q_0, \Delta_2, F_2)$.
Recall the algorithm \mbox{in~\cite[Lemma 30]{grobler2023remarks}} that decides the non-emptiness problem for reset~PA in $\NP$ by exploiting that $\omega$-languages accepted by a strong reset~PA are ultimately periodic~\cite[Lemma 29]{grobler2023remarks}.
%For any two states $p, q \in Q$ let $\Amc_{p \Rightarrow q} = (Q \cup \{q_0'\}, \Sigma, q_0', \Delta_{p \Rightarrow q}, \{q\}, C)$ where $\Delta_{p \Rightarrow q} = \{(q_1, a, \vbf, q_2) \mid (q_1, a, \vbf, q_2) \in \Delta, q_1 \notin F) \cup \{(q_0', a, \vbf, q_2) \mid (p, a, \vbf, q_2) \in \Delta\}$. 
%Hence, if we interpret $\Amc_{p\Rightarrow q}$ as a finite word~PA, then it accepts words all accepted by the finite word~PA $\Amc$ when starting in $p$, ending in $q$, and not visiting an accepting state in-between. In other words, if $p,q\in F$, then $\Amc_{p \Rightarrow q}$ accepts all finite infixes that the strong reset~PA $\Amc$ may read when the last reset was in $p$, and the next reset is in $q$. %Furthermore, this construction preserves determinism (up to completeness, but we can always complete the automaton by adding a non-accepting sink).
The algorithm guesses an accepting state $p_f$ of $\Amc_1$ that is seen infinitely often, and a sequence of distinct accepting states $p_1 \dots p_n$ such that $p_k = p_f$ for some $k \leq n$. 
Then, the algorithm tests for all $0 \leq i < n$ whether the finite word~PA $\Amc_{p_i \Rightarrow p_{i+1}}$ accepts at least one word, and whether $\Amc_{p_n \Rightarrow p_k}$ accept at least one word (where $\Amc_{p \Rightarrow q}$ is defined as in \Cref{lem:complSR}).
We modify the algorithm as follows.

First, let $\Amc$ be the product automaton of $\Amc_1$ and $\Bmc$ (as in the previous proof; however, the set of accepting states is not important for the algorithm).
We guess state $p_f \in F_1$ and $q_f \in F_2$ that we expect to be seen infinitely often to satisfy the acceptance conditions of~$\Amc_1$ resp.\ $\Bmc$. Furthermore, similar to the algorithm above we guess a sequence of distinct states $(p_1, q_1) (p_2, q_2) \dots (p_n, q_n)$ of $\Amc$ such that for some $\ell \leq n$ we have $q_\ell = q_f$, for some $k \leq \ell$ we have $p_k = p_f$, and for all $j \neq \ell$ we have $p_j \in F_1$.
If $p_\ell \in F_1$, we proceed exactly as in the original $\NP$-algorithm, that is, for every $0 \leq i < n$, we test the finite word PA $\Amc_{(p_i, q_i) \Rightarrow (p_{i+1}, q_{i+1})}$, as well as the finite word PA $\Amc_{(p_n, q_n) \Rightarrow (p_{k}, q_{k})}$ for non-emptiness. 

Now assume $p_\ell \notin F_1$.
Let $\Amc_{p_{\ell-1} \Rightarrow q_\ell \Rightarrow p_{\ell + 1}}$ be the finite word~PA that accepts all finite infixes accepted by the product automaton $\Amc$ when starting in $(p_{\ell-1}, q_{\ell-1})$, visiting $(p_\ell, q_\ell)$ at some point, and ending in $(p_{\ell+1}, q_{\ell+1})$ such that for all internal states $(p_i, q_i)$ we have $p_i \notin F_1$.
Two achieve this, we take two copies of $\Amc$, where all accepting states in the first copy are not reachable, all accepting states in the second copy have no outgoing transitions, and the second copy can only be reached from the first copy via $(p_\ell, q_\ell)$.

Formally, let 
$$\Amc_{p_{\ell-1} \Rightarrow q_\ell \Rightarrow p_{\ell + 1}} = (Q_1 \times Q_2 \times \{1,2\} \cup \{q_0'\}, \Sigma, q_0', \Delta', \{(p_{\ell+1}, q_{\ell+1})\}, C)$$ 
with 
\begin{align*}
    \Delta' =&\ \{((p,q,1), a, \vbf, (p', q', 1)) \mid ((p,q), a, \vbf, (p', q')) \in \Delta, p, p' \notin F_1\} \\
    \cup&\ \{((p,q,2), a, \vbf, (p', q', 2)) \mid ((p,q), a, \vbf, (p', q')) \in \Delta, p \notin F_1\} \\
    \cup&\ \{((p,q,1), a, \vbf, (p_\ell, q_\ell, 2) \mid ((p, q), a, \vbf, (p_\ell, q_\ell)) \in \Delta \} \\
    \cup&\ \{(q'_0, a, \vbf, (p', q', 1) \mid ((p_{\ell-1}, q_{\ell-1}), a, \vbf, (p', q' )) \in \Delta, p' \notin F_1 \} \\
    \cup&\ \{(q'_0, a, \vbf, (p_\ell, q_\ell, 2) \mid ((p_{\ell-1}, q_{\ell-1}), a, \vbf, (p_\ell, q_\ell)) \in \Delta \}.
\end{align*}

Now the algorithm is similar the first case with the addition that we also test this automaton for non-emptiness. Hence, for every $0 \leq j < \ell-1$ and $\ell+1 \leq j < n$ we test $\Amc_{(p_j, q_j) \Rightarrow (p_{j+1}, q_{j+1})}$ as well as $\Amc_{(p_n, q_n) \Rightarrow (p_k, q_k)}$ for non-emptiness. Furthermore, we test $\Amc_{p_{\ell-1} \Rightarrow q_\ell \Rightarrow p_{\ell + 1}}$ for non-emptiness. 

If all these tests succeed, we conclude $SR_\omega(\Amc_1) \cap L_\omega(\Bmc) \neq \varnothing$, as they witness the existence of an ultimately periodic word $\alpha$ accepted by $\Amc_1$ with the property that there is a run an accepting run of $\Bmc$ on $\alpha$ that visits $q_f$ infinitely often.
\end{proof}

Let $\Amc$ be a det.\ strong reset PA whose semi-linear set can by complemented in polynomial time. By \Cref{lem:complSR} we can compute a Büchi automaton $\overline{\Bmc}$ and a reach.\ PA $\Amc^\circ$ such that $\overline{SR_\omega}(\Amc) = L_\omega(\overline{\Bmc}) \cup R_\omega(\Amc^\circ)$.
By the previous lemma we can test $SR_\omega(\Amc) \cap L_\omega(\overline{\Bmc}) = \varnothing$ in \coNP.
Hence, it remains to show that testing intersection emptiness of a deterministic strong reset~PA and a reachability PA is decidable in $\coNP$. To achieve that, we use the \NP-algorithm in~\cite{zvassnz} deciding the reachability problem for $\Zbb$-$\mathsf{VASS}$ with $k$ nested zero-tests ($\Zbb$-$\mathsf{VASS}^\mathsf{nz}_k$).
A $\Zbb$-$\mathsf{VASS}^\mathsf{nz}_k$ (of dimension $d \geq 1$) is a tuple $V = (Q, Z, E)$ where $Q$ is a finite set of states, 
$Z \subseteq \{0, 1, \dots, d\}$ is its set of zero tests with $|Z \setminus \{0\}| = k$, and $E \subseteq Q \times \Zbb^d \times Z \times Q$ is a finite set of transitions. 
A configuration of $V$ is a pair $(p, \vbf) \in Q \times \Zbb^d$. Assume $\vbf = (v_1, \dots, v_d)$. We write $(p, \vbf) \vdash_V (p', \vbf')$ if there is a transition $(p, \ubf, \ell, p') \in E$ such that $\vbf' = \vbf + \ubf$ and $v_1 = \dots = v_\ell = 0$. Furthermore, we write $(p, \vbf) \vdash^*_V (p', \vbf')$ if there is a sequence $(p_1, \vbf_1) \vdash_V \dots \vdash_V (p_n, \vbf_n)$ for some $n \geq 1$ such that $(p, \vbf) = (p_1, \vbf_1)$ and $(p', \vbf') = (p_n, \vbf_n)$. 
Intuitively, a $\Zbb$-$\mathsf{VASS}^\mathsf{nz}_k$ is a counter machine with zero-tests that can only zero-test the top-most $\ell$ counters at once, for $k$ different values of $\ell$.

The \emph{reachability} problem for $\Zbb$-$\mathsf{VASS}^\mathsf{nz}_k$ is defined as follows: given a $\Zbb$-$\mathsf{VASS}^\mathsf{nz}_k$ $V$ and two configurations $(p, \vbf), (p', \vbf')$, does $(p, \vbf) \vdash_V^* (p', \vbf')$ hold? As shown in~\cite{zvassnz}, this problem is $\NP$-complete\footnote{We note that our definition of $\Zbb$-$\mathsf{VASS}^\mathsf{nz}_k$ differs slightly from the definition in~\cite{zvassnz}, as we allow $E\subseteq Q\times \Zbb^d\times Z\times Q$ instead of $E\subseteq Q\times \{-1,0,1\}^d\times Z \times Q$ only. However, this difference does not change the mentioned complexity for the reachability problem~\cite{zvassnz}, see also~\cite[Section A.1]{logspaceconversion}.} for any fixed~$k$.
We are now ready to show that universality and inclusion for det.\ strong reset PA are $\coNP$-complete if we can complement their semi-linear sets in polynomial time.

\begin{lemma}
Let $\Amc_1$ and $\Amc_2$ be deterministic strong reset PA with the guarantee that the semi-linear set of $\Amc_2$ can be complemented in polynomial time. Then the following questions are $\coNP$-complete.
\begin{enumerate}
    \item Is $SR_\omega(\Amc_2) = \Sigma^\omega$?
    \item Is $SR_\omega(\Amc_1) \subseteq SR_\omega(\Amc_2)$?
\end{enumerate}
\end{lemma}
\begin{proof}
Hardness follows again from a reduction from partition, very similar to the reduction in \Cref{lem:limitcoNP}. Hence, we focus on the containment in $\coNP$ of the second question.

As mentioned above, we can compute a Büchi automaton $\overline{\Bmc}$ and a reach.\ PA $\Amc^\circ$ such that $\overline{SR_\omega}(\Amc_2) = L_\omega(\overline{\Bmc}) \cup R_\omega(\Amc^\circ)$ by \Cref{lem:complSR}, and test $SR_\omega(\Amc_1) \cap L_\omega(\overline{\Bmc}) = \varnothing$ in $\coNP$ by \Cref{lem:interSRBukki}.

It remains to show how to solve $SR_\omega(\Amc_1) \cap R_\omega(\Amc^\circ) = \varnothing$ in $\coNP$.
In order to do so we use the algorithm in~\cite{zvassnz} solving reachability for $\Zbb$-$\mathsf{VASS}^\mathsf{nz}_2$ in $\NP$ to decide $SR_\omega(\Amc) \cap R_\omega(\Amc^\circ) \neq \varnothing$ in $\NP$. 
Let $\Amc_1 = (Q_1,\Sigma,p_0, \Delta_1, F_1, C_1)$ be of dimension $d_1$ and $\Amc^\circ = (Q_2,\Sigma, q_0, \Delta_2, F_2, C_2)$ be of dimension $d_2$.
Assume $SR_\omega(\Amc) \cap R_\omega(\Amc^\circ) \neq \varnothing$ and let $\alpha$ be an infinite word accepted by both automata. In particular, there is a finite prefix $u$ of $\alpha$ and an accepting run of $\Amc^\circ$ on $\alpha$ satisfying the Parikh condition after processing $u$, say in the accepting state $q_f \in F_2$. Recall that all outgoing transitions of every accepting state of $\Amc^\circ$ are self-loops, hence we may assume that $\Amc^\circ$ does not leave $q_f$ anymore after processing $u$.
Now let $p_f \in F_1$ be the first accepting state visited by the unique accepting run of $\Amc_1$ on $\alpha$ after processing $u$, say after processing the prefix $uv$. Note that $\Amc^\circ$ is still in $q_f$ after processing $uv$.

Our strategy is as follows.
First, we guess states $p_f$ and $q_f$ with the mentioned properties. We then verify these properties by building a product $\Zbb$-$\mathsf{VASS}^\mathsf{nz}_2$ with the property that $((p_0, q_0), \0) \vdash^*_V ((p_f, q_f), \0)$ if and only there is a finite prefix $uv$ as described above. Then we need to test whether $\Amc_1$ can continue a partial run from $p_f$ to an accepting run, that is, whether $\Amc_2$ with initial state $p_f$ accepts at least one infinite word, say $\beta$. If all these tests succeed, the infinite word $uv\beta$ witnesses non-emptiness.

We build a product $\Zbb$-$\mathsf{VASS}^\mathsf{nz}_2$ $V$ with $d_1 + d_2$ many counters. The idea is as follows. We use the first $d_1$ counters with to simulate $\Amc_1$ ensuring that every visit of an accepting state of $\Amc_1$ is with good counter values by zero-testing them.
Likewise, we use a second set of $d_2$ counters to simulate $\Amc^\circ$.
Let us give some more details on how to verify that every visit of an accepting state implies good counter values. Let $C_1 = C(\bbf_1, P_1) \cup \dots \cup C(\bbf_k, P_k)$ for some~$k \geq 1$. 
For every accepting state $f \in F_1$ we insert $k$ states, say $f^{(1)}, \dots, f^{(k)}$, and a copy of~$f$ itself. We connect~$f$ to $f^{(i)}$ with a transition subtracting $\bbf_i$ and no zero-test.
Then, for every period vector $\pbf_i \in P_i$, we insert a loop on~$f^{(i)}$ subtracting~$\pbf_i$. 
Finally, every outgoing transition of $f^{(i)}$ is equipped with a zero-test on the first $d_1$ counters. 
This construction allows us to test membership of the current counter values in $C_1$, while resetting the counters in parallel.
Finally, when reaching $(p_f, q_f)$, we use the same idea to check whether the vector induced by the last $d_2$ counters yields a vector in $C_2$. As $p_f$ is accepting, we expect all counters to be zero at this point, which we check by zero-testing them all.
Hence, if $((p_0, q_0), \0) \vdash^*_V ((p_f, q_f), \0)$ holds (where we assume that $(p_f, q_f)$ is the state of~$V$ reached after the described zero-tests), this implies that there is a finite prefix $uv$ such that~$\Amc_2$ rejects every infinite word with prefix $uv$ due to bad counter values, while the unique partial run of~$\Amc_1$ on $uv$ respects the strong reset acceptance condition.
Finally, we need to check whether the $\omega$-language of $\Amc_1$ with initial state $p_f$ is not-empty using the $\NP$-algorithm for testing emptiness for strong reset PA~\cite[Lemma 30]{grobler2023remarks}.
If these tests succeed, there is an infinite word $\beta$ such that $uv\beta \in SR_\omega(\Amc_1) \setminus SR_\omega(\Amc_2)$, witnessing non-inclusion.
\end{proof}

We are now ready to combine and generalize the results in the previous lemmas to all det.\ strong reset PA.

\begin{lemma}
Universality and inclusion for det.\ strong reset PA are $\Pi_2^\P$-complete.
\end{lemma}
\begin{proof}
    Hardness follows from~\Cref{cor:universalityLimit}; hence, we show that non-inclusion is in $\Sigma_2^\P$, yielding the desired result.

    Let $\Amc_1 = (Q_1, \Sigma, p_0, \Delta_1, F_1, C_1)$ and $\Amc_2 = (Q_2, \Sigma, q_0, \Delta_2, F_2, C_2)$ be det.\ strong reset~PA of dimensions $d_1$ and $d_2$, resp. 
    As in the previous lemma, we can compute a Büchi automaton~$\overline{\Bmc}$ accepting all infinite words that are rejected by $\Amc_2$ because every accepting state of $\Amc_2$ is visited only finitely often. 
    By \Cref{lem:interSRBukki}, we can test $SR_\omega(\Amc_1) \cap L_\omega(\overline{\Bmc}) \neq \varnothing$ in~$\NP$. 
    If the intersection is indeed non-empty, we conclude $SR_\omega(\Amc_1) \not\subseteq SR_\omega(\Amc_2)$.
    Otherwise, this non-inclusion might hold as there is an infinite word $\alpha \in SR_\omega(\Amc_1)$ rejected by $\Amc_2$ due to a reset with bad counter values. To test this case, we proceed as follows.
    
    We guess two accepting states $q_2$, $q_3$ of $\Amc_2$ and
    utilize the non-irrelevance algorithm in \Cref{lem:irrelevance} to test the existence of such a finite infix with bad counter values, \ie, an infix whose unique partial run of $\Amc_2$ yields a vector $\vbf \notin C_2$. We then construct a $\Zbb$-$\mathsf{VASS}^\mathsf{nz}_2$ to test whether there is an infinite word accepted by $\Amc_1$ that is rejected by $\Amc_2$ because the partial run between $q_2$ and $q_3$ yields the bad vector $\vbf$, ensuring that $\Amc_1$ respects all resets, similar to the proof of the previous lemma.

    Let $\Amc = (Q_1 \times, Q_2, \Sigma, (p_0, q_0), \Delta, F_1 \times Q_2, C_1)$ with \[\Delta = \{((p,q), a, \vbf, (p',q')) \mid (p,a,\vbf,p') \in \Delta_1, (q,a,\cdot,q') \in \Delta_2\}\] be the product automaton of $\Amc_1$ and $\Amc_2$ where we only keep the vectors and accepting states from $\Amc_1$. We guess two states $(p_2, q_2), (p_3, q_3) \in Q_1 \times F_2$ such that there is $\alpha \in SR_\omega(\Amc_1) \setminus SR_\omega(\Amc_2)$ with the property that $\Amc_2$ rejects $\alpha$ because the unique rejecting run of $\Amc_2$ on $\alpha$ visits $q_2$ at some position, say $f_2$ (and hence resets), the next reset is in $q_3$ at position $f_3$, and the vector collected in this time is not a member of $C_2$. Furthermore, $p_2$ and $p_3$ are the corresponding states of $\Amc_1$ at positions $f_2$ and $f_3$.
    Additionally, we guess two states $(p_1, q_1), (p_4, q_4) \in F_1 \times Q_2$ such that the unique accepting run of $\Amc_1$ on $\alpha$ resets the last time before reaching position $f_2$ say at position $f_1 \leq f_2$, namely in $p_1$; and $\Amc_1$ resets the first time after reaching position $f_3$ say at position $f_4 \geq f_3$, namely in $p_4$. Furthermore, $q_1$ and $q_4$ are the states of $\Amc_2$ at positions $f_1$ and $f_4$. Observe that $f_1 = f_2$ and $f_3 = f_4$ are possible.

    First, we test whether $(p_1, q_1)$ is reachable in $\Amc$ in the sense that there is a finite prefix $u$ of $\alpha$ such that $\Amc$ is in state $(p_1, q_1)$ after reading $v$ and $\Amc_1$ resets with good counter values whenever an accepting state of $\Amc_1$ is seen. In order to do so, we modify $\Amc$ such that $(p_1, q_1)$ is the only accepting state, and replace its outgoing transitions with trivial self-loops. Then we use the $\NP$-algorithm in~\cite[Lemma 30]{grobler2023remarks} to test non-emptiness of the resulting automaton.

    Second, we test whether it is possible to successfully continue the run from $(p_4, q_4)$ in the sense that there is an infinite word $\beta$ accepted by $\Amc$ when starting in $(p_4, q_4)$ (and hence by~$\Amc_1$ when starting in $p_4)$. In order to do so, we modify $\Amc$ such that $(p_4, q_4)$ is the initial state and test non-emptiness of the resulting automaton, again using the \NP-algorithm in~\cite[Lemma 30]{grobler2023remarks}.

    Let $\Amc'$ be defined as $\Amc$ but this time we only keep the vectors from $\Amc_2$ (instead of $\Amc_1$) and its semi-linear set is $C_2$ (instead of $C_1$).
    Third, we test whether there is indeed a finite prefix of $w$ of $\alpha$ such that $\Amc'$ is in state $(p_3, q_3)$ after processing $w$ when starting in $(p_2, q_2)$, and the vector collected by the unique partial run of $\Amc'$ (and hence $\Amc_2$) is not a member of~$C_2$.
    To achieve that we compute the PA $\Amc'_{(p_2, q_2) \Rightarrow (p_3, q_3)}$ and test for non-irrelevance using the algorithm in \Cref{lem:irrelevance} in $\Sigma_2^\P$ (observe that the construction of this PA preserves determinism up to completeness, and we can always complete the PA by adding a non-accepting sink). Recall that a part of this algorithm guesses a vector $\vbf$ not contained in~$C_2$ such that there is a run ending in the accepting state (here $(p_3, q_3)$) collecting $\vbf$. We need to keep this vector for the next step.

    Finally, we need to verify that there is a non-rejecting partial run of $\Amc$ on an infix~$vwx$, starting in $(p_1, q_1)$, visiting $(p_2, q_2)$ and $(p_3, q_3)$ in between, and ending in $(p_4, q_4)$ such that $w$ is processed between the visits of $(p_2, q_2)$ and $(p_3, q_3)$, witnessing that $\Amc_2$ is indeed rejecting. 
    To achieve that, we construct a product $\Zbb$-$\mathsf{VASS}^\mathsf{nz}_2$ $V$ of dimension $d_1 + d_2$ in a similar way as in the proof of the previous lemma. We give a high level description of $V$. The state set of~$V$ consists of three copies of the product $Q_1 \times Q_2$, and the transitions keep the vectors of the transitions of both automata, $\Amc_1$ and $\Amc_2$. 
    We can move from the first copy to the second copy upon reaching $(p_2, q_2)$, and from the second copy to the third copy  upon reaching~$(p_3, q_3)$. 
    The $d_2$ counters belonging to $\Amc_2$ are frozen (that is $\0$) in the first and third copy. Contrary, all counters are used in the second copy. Furthermore, we remove every state in $Q_1 \times F_2$ in the second copy besides $(p_2, q_2)$ and $(p_3, q_3)$.
    We verify that $\Amc_1$ always resets with good counter values using zero tests as in the proof of the previous lemma.
    Finally, upon reaching~$(p_4, v_4)$, we subtract $\0^{d_1} \cdot \vbf$ (where $\vbf$ is the vector guessed in the previous step) and test whether all counters are zero. As $p_4$ is accepting, we expect the first~$d_1$~counters to be zero. 
    Furthermore, as $\vbf$ is a vector witnessing bad counter values with respect to $\Amc_2$, all-zero counters imply that there is indeed such an infix $w$ of $\alpha$ breaking the run of $\Amc_2$. Let $(p_f, q_f)$ be the state of $V$ reached after this zero test.
    Then we use the $\NP$-algorithm in~\cite{zvassnz} to test $((p_1, q_1), \0) \vdash^*_V ((p_f, q_f), \0)$. If the answer is positive, we conclude that $SR_\omega(\Amc_1) \not\subseteq SR_\omega(\Amc_2)$, as witnessed by $\alpha = uvwxy\beta$.

    We conclude with a technical remark: the bit size of the vector $\vbf$ guessed in the third step might polynomially depend on $C_1$ and $C_2$. In particular, the bit size of (a suitable encoding of) $C_1$ might be arbitrary larger than the bit size of $C_2$. Hence, when calling the irrelevance algorithm with $\Amc'$ (where only $C_2$ is present) we need to take in account that the bit size of $\vbf$ might not be polynomial in $C_2$ but in $C_1$ and $C_2$.
\end{proof}

Finally, we study the intersection-emptiness problems, being the core of solving existential model checking. We begin with Büchi PA.
\begin{lemma}
    Intersection-emptiness for (deterministic and non-deterministic) Büchi PA is $\coNP$-complete.
\end{lemma}
\begin{proof}
    The lower bound for det.\ Büchi PA follows from their $\coNP$-complete emptiness problem.
    
    We give a proof sketch showing that intersection non-emptiness for Büchi PA is in $\NP$ by utilizing a recent result essentially stating that Ramsey-quantifiers in Presburger formulas can be eliminated in polynomial time~\cite{ramsey}. The authors show how to use the Ramsey-quantifier to check liveness properties for systems with counters. In particular, the existence of an accepting run of a Büchi~PA (answering the question whether the accepted $\omega$-language is non-empty) can be expressed with a Presburger formula with a Ramsey-quantifier.
    Hence, checking if the intersection of the two $\omega$-languages recognized by two Büchi~PA can be tested by intersecting two Presburger-formulas and moving the quantifiers to the front. We refer to Sections 4.1 and 8.2 in~\cite{ramsey} for more information.
\end{proof}

This settles also the case for limit PA and reach-reg.\ PA, as they are generalized by Büchi PA\footnote{In fact, the construction in~\cite{grobler2023remarks} turning limit PA into reach-reg.\ (which in turn are easily turned into Büchi PA) yields a reach-reg.\ PA whose size is exponential in the dimension of the limit PA. However, we can still utilize the presented \NP-algorithm testing intersection non-emptiness for Büchi PA by guessing a ``relevant'' part of polynomial size.}~\cite{grobler2023remarks}. The lower bound for their deterministic variants follows again from their $\coNP$-complete emptiness problems.
\begin{corollary}
    Intersection-emptiness for (deterministic and non-deterministic) limit PA and reach-reg.\ PA is $\coNP$-complete.
\end{corollary}

We conclude by showing that intersection emptiness is already undecidable for det.\ strong reset PA and hence also for det.\ weak reset PA.
The result relies on the fact that the intersection of two such languages can encode non-terminating computations of two-counter machines.
A two-counter machine $\Mmc$ is a finite sequence of instructions
\[(1 : \Isf_1)(2 : \Isf_2) \dots (k-1 : \Isf_{k-1})(k : \mathsf{STOP})\]
where the first component of a pair $(\ell, \Isf_\ell)$ is the line number, and the second component is the instruction in line $\ell$.
An instruction is of one of the following forms:

\begin{itemize}
    \item $\mathsf{Inc}(Z_i)$, where $i = 0$ or $i = 1$.
    \item $\mathsf{Dec}(Z_i)$, where $i = 0$ or $i = 1$.
    \item $\mathsf{If}\ Z_i = 0\ \mathsf{goto}\ \ell'\ \mathsf{else}\ \ell''$, where $i = 0$ or $i = 1$, and $\ell', \ell'' \leq k$.
\end{itemize}

Instructions of the first or second form are called increments resp.\ decrements, while instructions of the latter form are called zero-tests.
A configuration of $\Mmc$ is a tuple $c = (\ell, z_0, z_1)$, where $\ell \leq k$ is the current line number, and $z_0, z_1 \in \Nbb$ are the current counter values of $Z_0$ and $Z_1$ respectively. We say $c$ \emph{derives} into its unique successor configuration $c'$, written $c \vdash c'$, as follows.

\begin{itemize}
    \item If $\Isf_\ell = \mathsf{Inc}(Z_0)$, then $c' = (\ell + 1, z_0 + 1, z_1)$.
    \item If $\Isf_\ell = \mathsf{Inc}(Z_1)$, then $c' = (\ell + 1, z_0, z_1 + 1)$.
    \item If $\Isf_\ell = \mathsf{Dec}(Z_0)$, then $c' = (\ell + 1, \max\{z_0 - 1, 0\}, z_1)$.
    \item If $\Isf_\ell = \mathsf{Dec}(Z_1)$, then $c' = (\ell + 1, z_0, \max\{z_1 - 1, 0\})$.
    \item If $\Isf_\ell = \mathsf{If}\ Z_0 = 0\ \mathsf{goto}\ \ell'\ \mathsf{else}\ \ell''$, then $c' = (\ell', z_0, z_1)$ if $z_0 = 0$, and $c' = (\ell'', z_0, z_1)$ if $z_0 > 0$.
    \item If $\Isf_\ell = \mathsf{If}\ Z_1 = 0\ \mathsf{goto}\ \ell'\ \mathsf{else}\ \ell''$, then $c' = (\ell', z_0, z_1)$ if $z_1 = 0$, and $c' = (\ell'', z_0, z_1)$ if $z_1 > 0$.
    \item If $\Isf_\ell = \mathsf{STOP}$, then $c$ has no successor configuration.
\end{itemize}

The unique computation of $\Mmc$ is a finite or infinite sequence of configurations $c_0 c_1 c_2 \dots$ such that $c_0 = (1, 0, 0)$ and $c_i \vdash c_{i+1}$ for all $i \geq 0$.
Observe that the computation is finite if and only if the instruction $(k : \mathsf{STOP})$ is reached. If this is the case, we say $\Mmc$ terminates.
Given a two-counter machine $\Mmc$, it is undecidable to decide whether $\Mmc$ terminates~\cite{minsky}.

In the following we assume w.l.o.g.\ that our two-counter machines satisfy the guarded-decrement property~\cite{infiniteZimmermann}, which guarantees that every decrement does indeed change a counter value: every decrement $(\ell : \mathsf{Dec}(Z_i))$ is preceded by a zero-test of the form $(\ell - 1, \mathsf{If}\ Z_i = 0\ \mathsf{goto}\ \ell+1\ \mathsf{else}\ \ell)$.
Note that this modification does not change the termination behavior of a two-counter machine, as decrementing a counter whose value is already zero does not have an effect.

\begin{lemma}
    The intersection emptiness problem for deterministic strong reset~PA is undecidable.
\end{lemma}
\begin{proof}
We can encode infinite computations of two-counter machines as infinite words over $\Sigma = \{a,b, 1, 2, \dots, k\} \cup \Sigma_\Isf$, where $\Sigma_\Isf = \{I_a, I_b, D_a, D_b, Z_a, Z_b, \bar{Z}_a, \bar{Z}_b\}$. The idea is as follows. Let $c = (\ell, z_0, z_1)$ be a configuration of $\Mmc$. We encode $c$ as a finite word $w_c = \ell u x \in \Sigma^*$, where $\ell \in \{1,2, \dots, k\}$ encodes the current line number, $u \in \{a,b\}^*$ with $|u|_a = z_0$ and $|u|_b = z_1$ encodes the current counter values, and $x \in \Sigma_\Isf$ encodes the instruction $\Isf_\ell$ of line $\ell$ as follows: 

\begin{itemize}
    \item If $\Isf_\ell = \mathsf{Inc}(Z_0)$, then $x = I_a$, and if $\Isf_\ell = \mathsf{Inc}(Z_1)$, then $x = I_b$.
    \item If $\Isf_\ell = \mathsf{Dec}(Z_0)$, then $x = D_a$, and if $\Isf_\ell = \mathsf{Dec}(Z_1)$, then $x = D_b$.
    \item If $\Isf_\ell = \mathsf{If}\ Z_0 = 0\ \mathsf{goto}\ \ell'\ \mathsf{else}\ \ell''$, and the line number of the unique successor configuration of $c$ is $\ell'$, then $x = Z_a$ (that is, the zero-test is successful). Analogously with $x = Z_b$.
    \item If $\Isf_\ell = \mathsf{If}\ Z_0 = 0\ \mathsf{goto}\ \ell'\ \mathsf{else}\ \ell''$, and the line number of the unique successor configuration of $c$ is $\ell''$, then $x = \bar{Z}_a$ (that is, the zero-test fails). Analogously with $x = \bar{Z}_b$.
\end{itemize}

Let $w_c, w_{c'} \in \Sigma^*$ be two words encoding two configurations of~$\Mmc$. We call $w_c \cdot w_{c'}$ \emph{correct} if $c \vdash c'$. Hence, we can encode a unique infinite computations $c_0 c_1 c_2 \dots$ as an infinite word $w_{c_0} w_{c_1} w_{c_2} \dots$. We show that the $\omega$-language $L = \{w_{c_0} w_{c_1} w_{c_2} \dots \}$ can be written as the intersection of two deterministic strong reset~PA $\omega$-languages. Let
\begin{align*}
L_1 &= \{w_{c_0} w_{c_1} w_{c_2} \dots \mid w_{c_{2i}} w_{c_{2i+1}}\text{ is correct for every } i \geq 0\}, \text{ and} \\
L_2 &= \{w_{c_0} w_{c_1} w_{c_2} \dots \mid w_{c_{2i+1}} w_{c_{2i+2}} \text{ is correct for every } i \geq 0\}. 
\end{align*}

Observe that $L_1 \cap L_2 = L$, and $L$ is empty if and only if the unique computation of $\Mmc$ terminates.
Hence, it remains to show that $L_1$ and $L_2$ are recognized by deterministic strong reset~PA.
We argue for~$L_1$; the argument for $L_2$ is very similar.
The idea is as follows: We construct a deterministic strong reset~PA $\Amc_1$ with five counters that tests the correctness of two consecutive encodings of configurations, say $w_{c_{2i}} \cdot w_{c_{2i + 1}} = \ell_1 u_1 x_1 \cdot \ell_2 u_2 x_2$ with $\ell_1, \ell_2 \in \{1, 2, \dots, k\}$, $u_1 u_2 \in \{a,b\}^*$ and $x_1, x_2 \in \Sigma_\Isf$.
First observe that checking whether $\ell_2$ is indeed the correct line number (that is, the correct successor of~$\ell_1$) can be hard-coded into the state space of $\Amc_1$: if $x_1$ encodes an increment or decrement, we expect $\ell_2 = \ell_1 + 1$, and if $x_1$ encodes a successful or failing zero-test $\mathsf{If}\ Z_i = 0\ \mathsf{goto}\ \ell'\ \mathsf{else}\ \ell''$, we expect $\ell_2 = \ell'$ or $\ell_2 = \ell''$, respectively.
Four counters of $\Amc_1$ are used to count the numbers of $a$'s and $b$'s in $u_1$ and $u_2$, respectively. Then, if $x_1 = I_a$, we expect $|u_2|_a = |u_1|_a + 1$ and $|u_2|_b = |u_1|_b$, and so on.
To be able to perform the correct check, we also encode $x_1$ into the state space as well as the fifth counter by counting modulo~$|\Sigma_\Isf|$.
Observe that the guarded-decrement property ensures that decrements are handled correctly.
Hence, $\Amc_1$ has two sets of states counting the numbers of $a$'s and $b$'s of $u_1$ and $u_2$, accordingly, as well as a set of accepting states that is used to check the counter values. 

After such a check, $\Amc_1$ resets, and continues with the next two (encodings of) configurations. The automaton for $L_2$ works in the same way, but skips the first configuration. 
\end{proof}

\section{Conclusion}
\label{sec:conclusion}
We have studied the expressiveness, closure properties and classical decision problems of the deterministic variants of weak reset~PA, strong reset~PA, reach-reg.~PA, and most notably limit~PA. The latter are closed under the Boolean operations and hence all common decision problems are decidable for them, including the classical model checking problems. 
Closely related to model checking problems are synthesis problems. Here, the problem is to generate a model from a system specification (which is correct by construction). Gale-Stewart games play a key role in solving such synthesis problems~\cite{GaleStewart}. However, these games are undecidable when winning conditions are specified by automata whose emptiness or universality problem is undecidable. 
Our decidability results for deterministic limit PA raise the interesting and important question whether Gale-Stewart games can be solved when their winning condition is expressed by these automata. 

In future work we further plan to study the regular separability problem for these models, that is, given two $\omega$-languages $L_1, L_2$ recognized by~PA operating on infinite words, is there an $\omega$-regular language~$L$ with $L_1 \subseteq L$ and $L_2 \cap L = \varnothing$. 
Solving this problem can be used as an alternative approach to solving existential PA model checking, as (regular) separability implies intersection emptiness. It has already been studied for some related models, \eg PA on finite words~\cite{parikhsep} and Büchi $\mathsf{VASS}$~\cite{buechiVASSsep}.
Furthermore, it remains to classify the intersections of all incomparable models and thereby to provide a fine grained ``map of the universe'' for Parikh recognizable $\omega$-languages.

%\newpage
\bibliography{ref}

\end{document}